\documentclass[amsmath,onecolumn, floatfix,nofootinbib, prx, aps, usenames, a4paper, dvipsnames,10pt]{quantumarticle}
\pdfoutput=1

\usepackage[utf8]{inputenc}
\usepackage[english]{babel}
\usepackage[T1]{fontenc}
\usepackage{natbib}
\usepackage{bbold}

\usepackage[page]{totalcount}

\usepackage{qcircuit}

\usepackage[usenames,dvipsnames,table]{xcolor}
\usepackage{graphicx}
\usepackage[ colorlinks = true, 
             linkcolor = MidnightBlue,
             urlcolor  = MidnightBlue,
             citecolor = orange,
             anchorcolor = ForestGreen,
]{hyperref} 
\usepackage{float}

\usepackage[sans]{dsfont} 
\usepackage{bbm}
\usepackage[mathscr]{euscript}
\usepackage{cancel} 

\usepackage{xstring} 
\usepackage{xparse}

\usepackage{chngcntr}

\usepackage{tcolorbox} 
\tcbuselibrary{theorems} 
\usepackage{framed}
\usepackage{mdframed}

\usepackage{amsmath, amssymb, amsfonts, amsthm}
\usepackage{units}
\usepackage{mathtools}
\usepackage{thmtools, thm-restate} 

\usepackage{ucs}
\usepackage{subcaption}
\usepackage{fullpage}
\usepackage{csquotes}
\usepackage{epigraph}

\usepackage{cleveref}

\usepackage[normalem]{ulem}

\usepackage{enumerate}

\tcbset{colback=orange!5,colframe=orange!75!yellow, 
fonttitle= \bfseries, float=htb}

\newcommand{\footnoterecall}[1]{\hyperref[#1]{\footnotemark[\value{#1}]}}

\declaretheoremstyle[
    spaceabove=6pt, 
    spacebelow=6pt, 
    headfont=\normalfont\bfseries,
    notefont=\mdseries\bfseries, 
    notebraces={(}{)}, 
    bodyfont=\normalfont\itshape,
    postheadspace=1em,
    headpunct={:}]{thmstyle}

\newtheorem{theorem}{Theorem}[section]
\newtheorem{corollary}[theorem]{Corollary}

\newtheorem{lemma}[theorem]{Lemma}

\numberwithin{equation}{section} 
\counterwithin{figure}{section}
	\newcommand{\blue}[1]{\textcolor{Blue}{#1}}
	\newcommand{\green}[1]{\textcolor{OliveGreen}{#1}}

    \newcommand{\ket}[1]{\vert  #1 \rangle}
    \newcommand{\bra}[1]{\langle #1 |}
    \newcommand{\inprod}[2]{\langle #1 | #2 \rangle}
	\newcommand{\proj}[2]{\ket{#1}\bra{#2}}

	\newcommand{\id}{\mathbbm{1}}

	\newcommand{\tr}{\operatorname{Tr}  }
	
	\renewcommand{\vec}[1]{\mathbf{#1}}

	\newcommand{\avg}[1]{\langle #1 \rangle}

\usepackage{tikz}
\usetikzlibrary{arrows}
\usepackage{rotating}
\usetikzlibrary{shapes}
\usetikzlibrary{hobby}
\usetikzlibrary{decorations.markings}

\usepackage{multirow}
\usepackage{collcell}
\usepackage{multicol} 

\usetikzlibrary{tikzmark}

\newcommand{\ApplyGradient}[1]{%
        \pgfmathsetmacro{\PercentColor}{100}
        \hspace{-0.33em}\colorbox{cyan!\PercentColor!white}{}
}

\newcommand{\colorcell}[1]{%
        \hspace{0em}\colorbox{#1}{}
}

\newcolumntype{X}{>{\collectcell\colorcell}c<{\endcollectcell}}

\newcolumntype{R}{>{\collectcell\ApplyGradient}c<{\endcollectcell}}

\newcommand{\SetNumber}[1]{%
        \hspace{0.9em}#1\hspace{0.9em}
}
\newcommand{\SetNumberGrid}[1]{%
        \hspace{0.6em} #1 \hspace{0.6em}
}

\newcolumntype{C}{>{\collectcell\SetNumberGrid}c<{\endcollectcell}}
\newcolumntype{M}{>{\collectcell\SetNumber}c<{\endcollectcell}}

\newenvironment{homemadetoy}
    {\ \begin{tabular}
         {|{c}@{}|@{}{c}|{c}|{c}|}    
        \hline}
    {\\ \hline 
        \end{tabular}  \ }

\NewDocumentCommand{\lonecell}{o m o }
    {\ \scalebox{0.75}{\begin{tabular}{|c|} \hline\toycell[#1]{#2}[#3]\\ \hline \end{tabular}} \ }

\newcommand\DecodeToyColor[1]%
  {\IfStrEqCase{#1}%
        {1{\def\ProcessedArgument{cyan}}
         0{\def\ProcessedArgument{white}}
         g{\def\ProcessedArgument{ForestGreen!20}}
         r{\def\ProcessedArgument{red!20}}
        }[\def\ProcessedArgument{#1!20}]
  }

\newcommand{\DecodeToyPointer}[1]%
    {\IfStrEq{#1}{nil}%
        {\def\ProcessedArgument{}}%
        {\def\ProcessedArgument{\tikzmark{#1}}}%
   }

\NewDocumentCommand{\toycell}{o >{\DecodeToyColor}m o }
    {\scalebox{\IfValueTF{#1}{0.7}{1}}{\colorbox{#2}{\IfValueT{#1}{{\LARGE{#1}}}\IfValueT{#3}{\tikzmark{#3}}}}}

\newcommand{\customrow}[4]
    {\toycell{#1}&\toycell{#2}&\toycell{#3}&\toycell{#4}}

\NewDocumentCommand{\onecustomtoy}{>{\SplitArgument{3}{}} m}%
    {\begin{homemadetoy}
        \customrow #1
    \end{homemadetoy}
    }

\NewDocumentCommand{\twocustomtoys}{>{\SplitArgument{3}{}} m >{\SplitArgument{3}{}} m >{\SplitArgument{3}{}} m >{\SplitArgument{3}{}} m}{%
    \begin{homemadetoy}
        \customrow #1\\ \hline
        \customrow #2\\ \hline
        \customrow #3\\ \hline
        \customrow #4
    \end{homemadetoy}%
}

\newcommand{\toy}[1]{
\IfStrEqCase{#1}{
    {0}{\onecustomtoy{1100}}
    {1}{\onecustomtoy{0011}}
    {+}{\onecustomtoy{1010}}
    {-}{\onecustomtoy{0101}}
    {i}{\onecustomtoy{1001}}
    {-i}{\onecustomtoy{0110}}
    {mix}{\onecustomtoy{1111}}
    {nil}{\onecustomtoy{0000}}
    }[]%
}

\newcommand{\toys}[1]{
\IfStrEqCase{#1}{
    {00}{\twocustomtoys{0000}{0000}{1100}{1100}}
    {01}{\twocustomtoys{0000}{0000}{0011}{0011}}
    {10}{\twocustomtoys{1100}{1100}{0000}{0000}}
    {+0}{\twocustomtoys{0000}{1100}{0000}{1100}}
    {0+}{\twocustomtoys{0000}{0000}{1010}{1010}}
    {1+}{\twocustomtoys{1010}{1010}{0000}{0000}}
    {11}{\twocustomtoys{0011}{0011}{0000}{0000}}
    {bell}{\twocustomtoys{0001}{0010}{0100}{1000}}
    {-bell}{\twocustomtoys{1000}{0100}{0010}{0001}}
    {mixedleft}{\twocustomtoys{1100}{1100}{1100}{1100}}
    {mixedright}{\twocustomtoys{0011}{0011}{0011}{0011}}
    {mixedup}{\twocustomtoys{1111}{1111}{0000}{0000}}
    {mixedantidiagonal}{\twocustomtoys{1100}{1100}{0011}{0011}}
    {mixeddiagonal}{\twocustomtoys{0011}{0011}{1100}{1100}}
    {mix}{\twocustomtoys{1111}{1111}{1111}{1111}}
    {nil}{\twocustomtoys{0000}{0000}{0000}{0000}}
    }
    []
}

\newenvironment{toyplay}
    {\begin{tikzpicture}
        [remember picture, overlay]
    }
    {\end{tikzpicture}}

\newcommand{\permutarrow}[3][ForestGreen]{%
    \draw [#1, <->, line width=1.5pt] 
         ({pic cs:#2}) -- ({pic cs:#3});}

\newcommand{\toyCNOT}{%
        \begin{homemadetoy}
            \toycell{g}[firstin]&\toycell{Orange}[secondin]&\toycell{g}[firstout]&\toycell{orange}[fourthout]\\ \hline%
            \toycell{g}[thirdin]&\toycell{Orange}[fourthin]&\toycell{g}[thirdout]&\toycell{Orange}[secondout]\\ \hline%
            \toycell{0}&\toycell{blue}[fifthin]&\toycell0&\toycell{blue}[sixthin]\\ \hline%
            \toycell{0}&\toycell{blue}[fifthout]&\toycell0&\toycell{blue}[sixthout]
        \end{homemadetoy}
        \begin{toyplay}
        \permutarrow{firstin}{firstout}
        \permutarrow{thirdin}{thirdout}
        \permutarrow[Orange]{secondin}{secondout}
        \permutarrow[Orange]{fourthin}{fourthout}
        \permutarrow[blue]{fifthin}{fifthout}
        \permutarrow[blue]{sixthin}{sixthout}
    \end{toyplay}
}

\newcommand{\toymeasurementz}{
    \begin{homemadetoy}
        \toycell[0]{orange}&\toycell[0]{orange}&\toycell[1]{g}&\toycell[1]{g}
        \end{homemadetoy}
}

\newif\ifproofread

\newcommand{\changemarker}[1]{%
\ifproofread
\textcolor{OliveGreen}{#1}%
\else
#1%
\fi
}

\begin{document}

\proofreadtrue

\title{Toys can't play: physical agents in Spekkens' toy theory}

\author{Ladina Hausmann}
\affiliation{Institute for Theoretical Physics, ETH Zurich, 8093 Z\"{u}rich, Switzerland}
\email{hladina@phys.ethz.ch}

\author{Nuriya Nurgalieva}
\affiliation{Institute for Theoretical Physics, ETH Zurich, 8093 Z\"{u}rich, Switzerland}
\email{nuriya@phys.ethz.ch}

\author{Lídia del Rio}
\affiliation{Institute for Theoretical Physics, ETH Zurich, 8093 Z\"{u}rich, Switzerland}

\date{}

\begin{abstract}
\emph{Information is physical}~\cite{Landauer1961}, and for a physical theory to be universal, it should model observers as physical systems, with concrete memories where they store the information acquired through experiments and reasoning. Here we address these issues in Spekkens' toy theory~\cite{Spekkens_2005}, a non-contextual epistemically restricted model that partially mimics the behaviour of quantum mechanics. 
We propose a way to model physical implementations of agents, memories, measurements, conditional actions and information processing.   
We find that the actions of toy agents are severely limited: although there are non-orthogonal states in the theory, there is no way for physical agents to consciously prepare them.  Their memories are also constrained:  agents cannot forget in which of two arbitrary states a system is. 
Finally, we formalize the process of making inferences about other agents' experiments and model multi-agent experiments like Wigner's friend. 
Unlike quantum theory~\cite{Frauchiger_2018, Brukner2018, Nurgalieva_2019, Nurgalieva2021, Fraser2020} or box world~\cite{Vilasini_2019}, in toy theory there are no inconsistencies when physical agents reason about each other's knowledge.  
\end{abstract}

\maketitle


\setlength{\epigraphwidth}{4in}
\epigraph{There were realities the human mind was never meant to withstand, pressures it was never meant to survive. Knowledge is like the sea. \\ Go too deep, and the crushing weight of it could kill you.}{Seanan McGuire, \emph{Laughter at the Academy}}

{\setlength{\epigraphwidth}{3.8in}
\epigraph{You\dots are\dots a\dots toy!!!  You aren't the real Buzz Lightyear, you're an\dots aw, you're an action figure! You are a child's\dots plaything! }{\emph{Toy Story}}}
\newpage

\tableofcontents

\section{Introduction and summary of the toy theory}
\label{sec:introduction}
\paragraph{Physical information.} Physical theories that aim to describe the world at macroscopic scales should ideally be able to physically model observers, the experiments they perform, and their reasoning within the theory. 
For instance, the information acquired by agents must be stored in some physical form, in systems that we call memories. These can encompass biological brains but also classical and quantum computer memories, or even just measurement devices.
Processing that information is ultimately a physical process, and manipulations of an agent's memory don't simply result abstract epistemic changes of their knowledge, but  also in concrete physical changes.
For example, quantum measurement schemes~\cite{vonNeumann1955} show us that the process of acquiring information about a system entangles the system with our physical memory, and Landauer's principle~\cite{Landauer1961,Bennett1998} tells us that forgetting that information is to shuffle those correlations to the environment, at a thermodynamic cost. 

\paragraph{Abstract logic.}
Another feature that physical theories should satisfy is to allow for the information contained in the memories to be operated according to simple reasoning principles. Ideally we would like inferences such as “if Alice knows that $a$ is true, and she knows that $a$ implies $b$, then she knows that $b$ is true” to hold independently of the physical origins of $a$ and $b$.\footnote{For example, Alice knows that (a) it is raining, and (b) whenever it rains outside, Bob carries an umbrella with him. She concludes that Bob indeed is walking with an umbrella today. Here Alice is able to combine her knowledge of Bob's reasoning (b)) with her own observation (a) to infer Bob's behaviour. We will examine quantum examples later in the manuscript.} In other words, an abstract system of \emph{epistemic logic} should ideally be applicable to any physical setting within the theory~\cite{Nurgalieva_2019, Fraser2020}.  

\paragraph{Tension between abstract logic and physical information.} It was shown that for some theories where both of of these requirements are satisfied -- where agents reason about each other's knowledge and are themselves modeled as physical memories within the scope of the model -- experience inconsistencies. For example, this is the case for the quantum theory and generalized probability theories (in particular, so-called box world), where agents applying standard logic to reason about physical experiments can come to contradictory conclusions~\cite{Frauchiger_2018, Vilasini_2019}. 
Our ultimate goal is to understand which classes of theories exhibit similar incompatibility between multi-agent logic and physics. In previous work~\cite{Vilasini_2019}, we discuss the elements needed to define a reasoning agent for an arbitrary physical theory  (including physical descriptions of agents' memories and measurements, subjective state update rules, and how these interface with an abstract logic system). In this work we analyse these concepts specifically in Spekkens' toy theory.

\paragraph{Spekkens' toy theory.}
Spekkens' toy theory is an \emph{epistemically-restricted theory}~\cite{Spekkens_2005,Pusey_2012, SpekkensFoundations2016, Catani_2017,Coecke2011,Coecke2010,Backens2014,Comfort2021}. Such theories distinguish two types of states: ontic states, which encode the physical state of a system, and epistemic states -- the states of knowledge that an observer has about the system. The theory imposes restrictions on agents' knowledge (in Spekkens' case, an agent can only have access to half of the total information about the ontic state). The evolution of an ontic state is governed by the dynamics of underlying ontic theory, whereas the epistemic state is on top subject to particular rules, for example for updating after a measurement.  
For a detailed description of the formalism of the toy theory, we refer the reader to our review~\cite{Hausmann2021}.

\paragraph{Contribution and structure.} 
In Section~\ref{subsec:toy-summary}, we provide a minimal summary of the essence of Spekkens' toy theory. In Section~\ref{sec:learning}, we discuss the physical implementations of measurements in quantum theory and their analogues in the toy theory, and analyze restrictions on agents' choices in conditional preparation scenarios in the toy theory in Section~\ref{sec:acting}. In Section~\ref{sec:forgetting}, we consider the process of forgetting, and show that not all expressions of ignorance are valid in the toy theory. We discuss how agents can make inferences in Section~\ref{sec:reasoning}, and consider the consequences of our results for thought experiments like the Frauchiger-Renner setup~\cite{Frauchiger2018}, formulated for the case of the toy theory. We summarize our conclusions and outline future directions in Section~\ref{sec:conclusions}.
To keep the main manuscript light, we present only intuitive examples in low dimensions;  our general results hold for arbitrary dimensions, and are formalized in the appendix.

\subsection{Minimal summary of the Spekkens' toy theory}
\label{subsec:toy-summary}

\paragraph{Knowledge balance principle.}
Spekkens' toy theory is an \textbf{epistemically-restricted theory}~\cite{Spekkens_2005}. Such theories distinguish two types of states: \textbf{ontic states}, which encode the physical state of a system, and \textbf{epistemic states} -- the states of knowledge that an observer has about the system. 
In the toy theory, epistemic states  are constrained by the 
 \emph{knowledge balance principle}, inspired by the Heisenberg uncertainty principle:
            \begin{displayquote}
                ``If one has maximal knowledge, then for every system,
                at every time, the amount of knowledge one possesses
                about the ontic state of the system at that time must
                equal the amount of knowledge one lacks.'' \cite{Spekkens07}
            \end{displayquote}
{In practice, suppose that a system can be in one of  $2d$ possible ontic states, that is, specifying the value of $\log_2 (2d)=1+ \log_2 d$ degrees of freedom is sufficient to identify each ontic state. Then an observer is only allowed to access at most  $\log_2 d$ bits of information, and we say that the system \emph{has dimension} $d$; this and other constraints will become clearer with the examples ahead. }

\paragraph{System dimensions.} {Our results apply to the} general case where a toy system can be decomposed into $N$ subsystems of dimension $d$ each. The continuous limit, analogous to a particle moving in space, corresponds to $d \to \infty $. For example, the elementary systems analogous to qubits (toy bits) have $d=2$.  In the following we stick to intuitive visualizations for small discrete dimensions (one or two toy bits); we review the general case in Appendix~\ref{appendix:toyformalism}.

\paragraph{{Notation for toy bits.}} {We can represent up to two discrete toy systems via simple grid diagrams. For higher dimensions, grid diagrams aren't as useful; the interested reader may find a review of the necessary mathematical notation in the appendix.
A single system with two degrees of freedom ($d=2)$ has four different ontic states (labelled for example 1,2,3,4). To see this, note that answering two binary questions would be sufficient to identify the ontic state of the system, for example ``is the ontic state odd?'' and ``is the ontic state smaller than 3?''.
Each ontic state is represented by one of four boxes,}
\scalebox{0.7}{\begin{homemadetoy}%
    \toycell[1]0&\toycell[2]0&%
    \toycell[3]0&\toycell[4]0
\end{homemadetoy}}\!.
{To represent an observer's knowledge about the ontic state --- that is the epistemic state from their perspective --- we colour in some of the boxes. For example, $\{1,2\} =$}
\scalebox{0.7}{\begin{homemadetoy}%
    \toycell[\textcolor{white}{\textbf1}]1&\toycell[\textcolor{white}{\textbf2}]1&%
    \toycell[3]0&\toycell[4]0
\end{homemadetoy}}%
{represents ``the observer knows that the system's ontic state is either $o=1$ or $o=2$.''  For simplicity we usually omit the number labels. 
One can draw an analogy between toy states and quantum states, }
\begin{align*}
    \ket0 \sim &\toy0 =\{1,2\}, &\  \ket+ \sim &\toy+  =\{1,3\},& \ \ket{+i} \sim &\toy i =\{1,4\}, \\
    \ket1 \sim& \toy1  =\{3,4\}, & \  \ket- \sim &\toy-  =\{2,4\}, &  \ket{-i} \sim &\toy{-i} =\{2,3\}.
\end{align*}
{These are the only pure states at $d=2$: they are states of maximal allowed information, according to the knowledge balance principle, for which the observer knows half of the degrees of freedom (for example \scalebox{0.7}{\toy+}\! represents ``the ontic state is odd'').  The quantum analogy  carries through to a stabilizer formulation of the toy theory \cite{Pusey_2012}.
The fully mixed state (a state of maximal ignorance) is represented as the mixture}
\begin{align*}
   \{1,2,3,4\}= \toy{mix} 
        &= \toy0 \vee \toy1  \sim  \proj00+\proj11 \\
        &= \toy+ \vee \toy- \sim  \proj+++\proj-- \\
        &= \toy{i} \vee \toy{-i} \sim  \proj {\,i\,}{\, i\,}+\proj{-i}{-i}.
\end{align*}
{ Unlike quantum theory, this is the only physical mixed state at $d=1$, that is the only mixed state that can emerge as a marginal of a globally pure state (like the entangled state which we will consider in a moment). The epistemic restriction implies that for a system composed of $N$ elementary systems, an agent is only allowed to have access to exactly $0\leq J\leq N$ bits of information (corresponding to an epistemic state spanning $2^{2N-J}$ ontic states). For example, for $N=1$ the valid epistemic states can span either 2 or 4 ontic states. The mixture}
\begin{align*}
    \begin{homemadetoy}
        \toycell{red}&\toycell{red}&\toycell{red}&\toycell0
    \end{homemadetoy}= \toy0 \vee \toy+
\end{align*}
{is not a valid epistemic state \cite{Spekkens07}.}\footnote{{At the time of writing (late 2022),  Spekkens and colleagues are investigating how to articulate a layer of (Bayesian) probabilistic knowledge on top of these epistemic ``physical'' states. If successful, this would allow for Bayesian states whose measurement statistics are identical to the illegal ``physical'' epistemic state \scalebox{0.7}{\begin{homemadetoy}\toycell{red}&\toycell{red}&\toycell{red}&\toycell0\end{homemadetoy}}, at least for local measurements. We leave it as future work to study the consequences and stability of that approach.}}

\paragraph{{Composing systems.}} {When we consider two toy systems, we represent the ontic states with a $4\times4$ grid, where the rows determine the possible ontic states of  system $A$ and the columns those of  system $B$, for example}
\begin{align*}
\begin{tabular}{c c c}
     $A$ \ & \toys{+0} 
        & $={\toy+}_A \otimes {\toy0}_B \ 
        \sim  \ \ket+_A \otimes \ket1_B$, \\ 
         \vspace{2mm} \\
     &  $B$ \\
    \vspace{5mm} \\
   $A$\ &\toys{mixedup} & $={\toy{1}}_A \otimes {\toy{mix}}_B \sim \proj11_A \otimes \id_B$ . \\
    \vspace{2mm} \\
     &  $B$ 
\end{tabular} 
\end{align*}
{We omit the subsystem labels when they are clear from context.
There are also global toy states which are not product states, for example the classically correlated state}
\begin{align}\label{toy:bell}
 \toys{mixeddiagonal} & = \toys{00} \vee \toys{11} \sim \proj{00}{00} + \proj{11}{11} , 
\end{align}
{which is a mixed state, and the pure entangled state }
\begin{align*}
\toys{bell} &\sim \ket{00} + \ket{11} .
\end{align*}
{Analogously to a Bell state in quantum theory, this latter state is used for toy teleportation and dense coding protocols \cite{Spekkens07}.  
Like their quantum analogues, entangled toy states are globally pure states with mixed marginals. In this example, the observer has maximal information about the correlations between $A$ and $B$, and maximal ignorance about the reduced state of individual subsystems. }

\paragraph{Reduced states.} { For discrete toy systems, taking the reduced state over one system corresponds to projecting the grid diagram into one axis. For example,}
\begin{align*}
    \begin{tabular}{c c}
         $E_{AB}= \quad A$ \ & \toys{+0}  \vspace{2mm} \\
         &  $B$ 
    \end{tabular} 
    \implies 
    E_{B}= 
     {\begin{homemadetoy}
         \customrow0000\\ 
         \customrow1100\\
         \customrow0000\\
         \customrow1100
     \end{homemadetoy}}_B
    ={\toy0}_B
\end{align*}
{More formally, if the global epistemic state is $E_{AB}$, the reduced (or marginal)  state on system $B$ is $E_B= \{o_B: \  (o_A, o_B) \in E_{AB} \}$. In the above example, $E_{AB} = \{(1,1), (1,2), (3,1), (3,2) \}$ so $E_B= \{1,2\}$. }

\paragraph{{Notation for transformations.}}
{For discrete dimensions, the allowed toy transformations are permutations of ontic states, constrained to mapping valid epistemic states to valid epistemic states (in all subsystems). For example consider 
the transformation that permutes the second and third ontic states,}
\begin{align*}
    H= \begin{homemadetoy}%
            \toycell0&\toycell{g}[leftend]&\toycell{g}[rightend]&\toycell0
        \end{homemadetoy} = (2,3).
    \begin{toyplay}
        \permutarrow{leftend}{rightend}
    \end{toyplay}
\end{align*}
{This permutation acts analogously to the Hadamard gate in quantum theory,  }
\begin{align*} 
    \ket0 \sim    \toy0 &\xleftrightarrow{H} \toy+ \sim \ket+, \\ 
     \ket1 \sim \toy1 &\xleftrightarrow{H}  \toy- \sim\ket-.
\end{align*}
{We will see other examples (like a CNOT gate) further ahead; in particular we will see that a controlled Hadamard is not a valid operation, and we will explore the implications of this for agents' free choice. }

\paragraph{{Notation for measurements.}}
{
In the toy theory, a measurement is a partition of the ontic state space into valid epistemic   states, called the measurement basis. The outcome of the measurement is determined by the position of the ontic state. The observer can then update their knowledge; in the toy theory the measurement update rule leads to an ontic disturbance, which we will later see emerges from the physical implementation of a measurement. The consequence for the updated epistemic state is a bit cumbersome to explain, but will become clear from examples right ahead. In short, the updated description should guarantee that if the observer repeats the measurement they obtain the same outcome, and be compatible with their overall knowledge \cite{Spekkens07}. 
 For example, consider the measurement 
$\mathcal M_Z= \scalebox{0.7}{\toymeasurementz}$\!, where the numbered partitions correspond to the epistemic states of the measurement basis. This is analogous to a Pauli-Z measurement of a single qubit. Suppose that Bob measures a toy bit in state $\scalebox{0.7}{\toy+}$ in this basis, and obtains outcome $\scalebox{0.8}{\lonecell[0]{orange}}$. This allows him to deduce that the ontic state previous to the measurement was $0$; however }
\begin{align*}
\begin{homemadetoy}%
    \toycell{violet}&\toycell0&\toycell0&\toycell0
\end{homemadetoy} = \toy+ \wedge \begin{homemadetoy}%
    \toycell{orange}&\toycell{orange}&\toycell0&\toycell0
\end{homemadetoy}
\end{align*}
{is not a valid epistemic state. Rather, it is a \textbf{pre- and post-selected state}, a concept that has analogues in quantum theory. The smallest epistemic state compatible with his knowledge of the pre- and post-selection  and with the requirement for repeated outcomes is \scalebox{0.7}{\toy0}. To see this, first note that there are four epistemic states compatible with Bob's knowledge, }
\begin{align*}
\begin{homemadetoy}%
\toycell{violet}&\toycell0&\toycell0&\toycell0
\end{homemadetoy} \subseteq  
\toy0, \toy+, \toy i, \toy{mix}.
\end{align*}
{The requirement for repeated outcomes translates to ``the post-measurment description must be a subset of $\scalebox{0.7}{\begin{homemadetoy}%
    \toycell{orange}&\toycell{orange}&\toycell0&\toycell0
\end{homemadetoy}}$, so that if I apply the same measurement again, I am guaranteed to obtain the same outcome \lonecell[0]{orange}.'' Of the four candidates, only the first epistemic state satisfies the condition,}
\begin{align*}
    \begin{homemadetoy}%
    \toycell{violet}&\toycell0&\toycell0&\toycell0
    \end{homemadetoy} 
    \subseteq  
\toy0 
    \subseteq   
    \begin{homemadetoy}%
    \toycell{orange}&\toycell{orange}&\toycell0&\toycell0
\end{homemadetoy}
\end{align*}
{This is analogous to the quantum measurement update rule for projective  outcomes: if Bob measured $\ket+$ in the $Z$ basis and obtained outcome 0, he would henceforth describe the state as $\ket0$. Now suppose that Bob is making the same measurement, but now on his half of the entangled toy Bell state (\ref{toy:bell}).  The global view of Bob's local measurement is }
\begin{align*}
\mathcal I \otimes \mathcal M_Z = 
\begin{homemadetoy}%
    \toycell[0]{orange}&\toycell[0]{orange}&\toycell[1]{g}&\toycell[1]{g}\\ \hline
    \toycell[0]{orange}&\toycell[0]{orange}&\toycell[1]{g}&\toycell[1]{g}\\ \hline
    \toycell[0]{orange}&\toycell[0]{orange}&\toycell[1]{g}&\toycell[1]{g}\\ \hline
    \toycell[0]{orange}&\toycell[0]{orange}&\toycell[1]{g}&\toycell[1]{g}
\end{homemadetoy}
\end{align*}
{If Bob obtains outcome \lonecell[0]{orange} he can deduce that the global ontic state before the measurement was one of  the two coloured squares in the bottom-left corner,}
\begin{align*}
    \begin{homemadetoy}%
        \toycell0&\toycell0&\toycell0&\toycell0\\ \hline 
        \toycell0&\toycell0&\toycell0&\toycell0\\ \hline 
        \toycell0&\toycell{violet}&\toycell0&\toycell0\\ \hline 
        \toycell{violet}&\toycell0&\toycell0&\toycell0
    \end{homemadetoy}
    = \toys{bell} \wedge 
    \begin{homemadetoy}%
        \toycell{orange}&\toycell{orange}&\toycell0&\toycell0\\ \hline 
        \toycell{orange}&\toycell{orange}&\toycell0&\toycell0\\ \hline 
        \toycell{orange}&\toycell{orange}&\toycell0&\toycell0\\ \hline 
        \toycell{orange}&\toycell{orange}&\toycell0&\toycell0
    \end{homemadetoy} .
\end{align*}
{This is the pre- and post-selected state that describes his knowledge of the ontic state of $AB$ just before the measurement. There are two valid epistemic states that are  compatible with this knowledge and live in the measurement partition \lonecell[0]{orange}, }
\begin{align*}
 \begin{homemadetoy}%
        \toycell0&\toycell0&\toycell0&\toycell0\\ \hline 
        \toycell0&\toycell0&\toycell0&\toycell0\\ \hline 
        \toycell0&\toycell{violet}&\toycell0&\toycell0\\ \hline 
        \toycell{violet}&\toycell0&\toycell0&\toycell0
    \end{homemadetoy}
    \subseteq
    \toys{00} ,\quad  &  \toys{mixedleft}
    \subseteq 
    \begin{homemadetoy}%
        \toycell{orange}&\toycell{orange}&\toycell0&\toycell0\\ \hline 
        \toycell{orange}&\toycell{orange}&\toycell0&\toycell0\\ \hline 
        \toycell{orange}&\toycell{orange}&\toycell0&\toycell0\\ \hline 
        \toycell{orange}&\toycell{orange}&\toycell0&\toycell0
    \end{homemadetoy} . \\
     \sim\proj00_A \otimes \proj00_B \quad  &  \sim \id_A \otimes \proj00_B
\end{align*}
{Of the two candidates, Bob picks the leftmost (the smallest state), which is the description that makes the most use of his knowledge about the state before the measurement (picking the state on the right would be discarding what he knew of the correlations between the two systems). 
The quantum analogy is when Bob makes a local $Z$ measurement on a Bell state $\propto \ket{00} + \ket{11}$, and upon seeing outcome $0$, updates his global description of the state to $\ket{00}$.}

\section{Learning: physical implementations of measurements}
\label{sec:learning}
\paragraph{Minimal settings.} 
In theories where agents can improve their knowledge about the state of the system by posing questions or performing measurements on the system, they have to update their epistemic state as a consequence. Hence, if we model agents as physical systems, the bare minimum of their description has to include the degree of freedom of the memory entry corresponding to the system being inquired by the agent. We call the process of revising the memory entry after a measurement a \textit{memory update}.
In this section, we consider how one can model the process of learning -- measuring a system and registering  the outcome in a memory -- in the toy theory.  It is in these settings that we find the most dramatic differences between quantum and toy theory. At every step, we start by reviewing the quantum process, and then try to find an analogous in the toy theory.

\begin{figure}[t!]
\centering
    \begin{subfigure}{\textwidth}
         \begin{align*}
        \Qcircuit @C=1.2em @R=1.2em {
        \lstick{\ket{\phi}_S} & \multigate{2}{V= e^{-\frac {i tg } {\hbar} \ \hat A_S \otimes \hat B_M}} & \qw  \\
        &&& \rstick{\sum_k \alpha_k \ket{a_k}_S \ket{\psi_k}_M} \\
        \lstick{\ket{\psi_0}_M} & \ghost{V= e^{-\frac {i tg } {\hbar} \ \hat A_S \otimes \hat B_M}} & \qw \gategroup{1}{3}{3}{3}{1em}{\}} \\
        }
        \end{align*}
        \caption{\textbf{Von Neumann measurement scheme~\cite{vonNeumann1955}.} An observable $\hat A_S = \sum_k a_k \ket{a_k}\bra{a_k}$ is measured by coupling the system of interest $S$ to a pointer $M$ through the Hamiltonian $\hat H= g\ \hat A_S \otimes \hat B_M$. If the original state of the system is $\ket\phi_S = \sum_k \alpha_k \ket{a_k}_S$, then the final state is  $\sum_k \alpha_k \ket{a_k}_S \otimes \ket{\psi_k}_M$, where $\ket{\psi_k}_M = e^{\frac{-i t g a_k}\hbar \hat B_M} \ket{\psi_0}_M$.
        }
    \end{subfigure} \\
    \begin{subfigure}{0.45\textwidth}
        \begin{align*}
        \Qcircuit @C=1.2em @R=1.2em {
        \lstick{\alpha \ket 0_S + \beta \ket 1_S} & \ctrl{2} & \qw  \\
        &&& \rstick{\alpha \ket{00}_{SM} + \beta \ket{11}_{SM}} \\
        \lstick{\ket 0_M} & \targ & \qw \gategroup{1}{3}{3}{3}{1em}{\}} \\
        }
        \end{align*}
        \caption{{ \bf  Discrete quantum measurement.} A qubit $S$ is measured in computational basis: it is coupled to another qubit $M$, acting as a memory system, by a CNOT gate. The operation coherently copies the state of $S$ to $M$.}
        \label{fig:discrete-cnot}
    \end{subfigure} \quad
    \begin{subfigure}{0.5\textwidth}
    \centering 
    \includegraphics[scale=0.23]{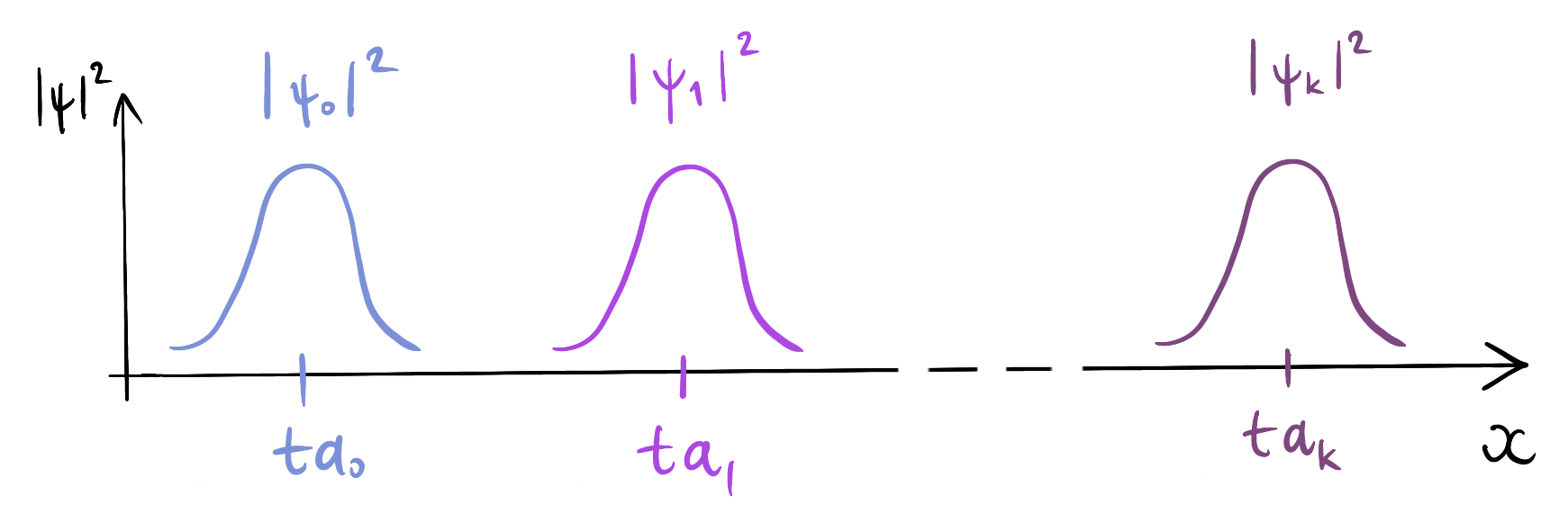}
        \caption{{ \bf Continuous quantum measurement.}  
        If $M$ is a continuous system and $\hat B_M = \hat P_M$ is the momentum operator and  then  each $\ket{\psi_k}_M$ has the original wave function of the pointer shifted by $a_k$, $\psi_k(x) = \psi_0(x- t a_k)$. For strong measurements, these peaks don't overlap.}
        \label{fig:continuous-cnot}
    \end{subfigure} 
\caption{{\bf Quantum measurements as physical evolutions.}  In quantum theory, we can implement measurements as unitary processes, which entangle the state of the system measured to the measurement device in the basis of the observable. Here are two familiar examples. Zooming out to the perspective of an external agent who sees the experimenter as just another quantum system, we can include the experimenter's memory and lab as part of the measurement device, and use the same unitary view to model the whole measurement process.}
\label{fig:measurement}
\end{figure}

\paragraph{Quantum measurements as entangling operations.} In quantum mechanics, measurements are entangling processes between the system of interest and a measurement device. For example, in the Stern-Gerlach experiment, the internal spin of a particle is coupled to its position degree of freedom through the application of a magnetic field; it is the position of the particle that acts as a pointer or measurement device when it hits a screen, as the arrival position is correlated with the internal spin.  

\paragraph{Quantum von Neumann measurement scheme.} More concretely, we can measure a discrete observable $\hat A_S = \sum_k a_k \proj{a_k}{a_k}_S$ on a system $S$ by coupling it to a measurement device (or pointer) $M$, through a Hamiltonian of the form $H = g\ \hat A_S \otimes \hat B_M$, where $\hat B_M$ is a suitable observable on $M$ and $g$ a tunable constant~\cite{vonNeumann1955}. 
Letting the two systems evolve for some time $t$ corresponds to the reversible evolution modelled by the unitary $V=  e^{\frac{-i  t} \hbar H } $. 
If $S$ is initially in an arbitrary state $\ket\phi_S = \sum_k \alpha_k \ket{a_k}_S$ {(with unknown coefficients $\alpha_k=\langle a_k | \phi \rangle_S$)}, and we prepare the measurement device (or  ``pointer'') in state $\ket{\psi_0}_M$, then after time $t$ the  global state becomes 
\begin{align*}
    e^{\frac{-i  t} \hbar H } \left( \ket\phi_S \otimes \ket{\psi_0}_M \right)
    =
    e^{\frac{-i g t} \hbar \ \hat A_S \otimes \hat B_M} \ \sum_k \alpha_k \ket{a_k}_S \otimes \ket{\psi_0}_M
    = \sum_k \alpha_k \ket{a_k}_S \otimes \underbrace{e^{\frac {-i t g a_k} \hbar \hat B_M} \ket{\psi_0}_M  }_{ \ket{\psi_k}_M}.
\end{align*}
This simple unitary evolution can be modified to account for noise, finite-size effects, coarse and continuous observables and other corrections, in order to cover realistic implementations of quantum measurements. A relevant remark for later is that ultimately, the observer's memory is  itself a quantum system that becomes entangled with the system measured --- at least from the perspective of an external agent.  

\paragraph{Quantum examples.}
Two examples for continuous and discrete measurements are summarized in Figure~\ref{fig:measurement}. The simplest case is when both the system to measured and the pointer are single qubits (Figure~\ref{fig:discrete-cnot}). An entangling CNOT gate\footnote{An interaction Hamiltonian that implements this gate in a suitable amount of time (e.g. $t=\pi$) is for example $H_\text{int} = \frac 1 4 Z'_S \otimes X'_M$, where $Z'_S$ and $X'_M$ are shifted $Z$ and $X$ operators, $Z'_S = Z_S - \id_S$,  $X'_M = X_M - 3 \id_M$.} between $S$ and $M$ implements a strong measurement of $S$ in the $Z$ basis: 
\begin{align}
   (\alpha \ket0_S + \beta \ket1_S) \otimes \ket0_M  \quad  \underrightarrow{\quad V= \operatorname{\textsc{cnot}} \quad } \quad  \alpha \ket0_S\ket0_M + \beta \ket1_S\ket1_M 
\end{align}
A familiar example in continuous systems is the position measurement of a particle, which entangles the particle and pointer in the position basis. This is achieved for example by setting $\hat A_S = \hat X_S$ and $\hat B_M = \hat P_M$, and initializing the pointer to a well-localized state (like a Gaussian wave\footnote{This measurement can be made sharper or weaker by tuning the parameters of the initial wave function of the pointer and the interaction time.} ). The measurement process results in the physical evolution
\begin{align*} 
    \left(\int dx \ \phi(x)\ \ket x_S\right) \otimes  \left(\int dx'\ \psi_0(x')\ \ket {x'}_M\right)  
     \quad  \underrightarrow{\quad V \quad } \quad 
    \int dx  \int d x'\  \phi(x)\  \psi_0(x'-t g x )\  \ket x_S \ket{x'}_M.
\end{align*}

\paragraph{Emergence of the post-measurement state.} {In quantum theory, the post-measurement state of a system $S$ emerges from the physical picture by moving the Heisenberg cut one level up and thinking of a projective measurement on the agent's memory  $M$ after the physical entangling operation between $S$ and $M$. For example, the (unnormalized) post-measurement state  when obtaining outcome $1$ on a $Z$ measurement of a qubit in state $\ket\psi_S$ can be expressed as }
\begin{align*}
    \tr_M \left[ \underbrace{\green{\id_S \otimes \proj11_M}}_{\green{\text{measuring memory}}} \ \underbrace{\blue{\operatorname{CNOT}} \  (\ \proj \psi\psi_S \  \blue{\otimes\  \proj00_M}) \ \blue{\operatorname{CNOT}}}_{\blue{\text{entangling $S$ and memory}}}   \right] =  \underbrace{\green{\proj11_S}}_{\green{\text{measuring $S$}}} \proj\psi\psi_S.
\end{align*}
{For details on the general case we refer to Appendix~\ref{appendix:quantum-analogues}. }

\paragraph{Measurements in the toy theory: discrete example.}
{We can now try to bring over the same concept of physical measurement to the toy theory.} Let us consider two separate toy bit systems: the first one is the system we measure, the second is the memory where the outcome is registered. 
We can proceed analogously to the quantum case, as there is an allowed transformation corresponding to the CNOT gate in the original toy theory~\cite{Spekkens_2005}. In the notation of grid diagrams it permutes the ontic states as
\begin{align*}
\operatorname{CNOT}= \toyCNOT
\end{align*}
Analogously to the quantum CNOT gate, the toy CNOT transformation  correlates two toy-bits, for example
\begin{align*}
   \ket+\ket0 \  \sim & \toys{+0} \xrightarrow{CNOT}  \toys{bell} \sim  \ket0\ket0+\ket1\ket1 \  , \\  
   \ket1\ket0 \  \sim  &\toys{10} \xrightarrow{CNOT} \toys{11}  \sim \ket1\ket1.  
\end{align*}
Once again, the post-measurement state of the system emerges from conditioning on a projective measurement on the memory, and then finding the reduced state.

\paragraph{Measurements in the toy theory: general case.} For the general case of arbitrary continuous or discrete dimensions, we can also find global transformations on the system and pointer that implement any valid measurements (Theorem~\ref{thm:coherentcopyobservable}). The transformation is essentially  analogous to the von Neumann measurement procedure, correlating the two toy subsystems in a way that reflects the properties of the observable, and recovers the post-measurement state of the system measured when we condition on the outcome. There are a few extra constraints (for example, on how the dimensions of system and pointer should match), but the main qualitative difference to quantum measurements comes from the fact that toy observables are already restricted at the abstract level, as we saw in the introduction. See Appendix~\ref{appendix:measproofs} for examples in continuous dimensions (like a position measurement) and detailed derivations of the formal results.

\section{Acting: restrictions on agents' choices} 
\label{sec:acting}

\paragraph{Choices start with measurements.} If we model agents' actions as physical processes, we see that they can always be decomposed as a measurement followed by a conditional transformation. The process of deciding which of a series of actions to take is ultimately a measurement --- of one's memory, of a randomness generated, or any relevant external systems. We look at the weather to decide what to wear, consult our agenda to decide on appointments, and even when making random decisions we can model our source of randomness as an explicit physical system. Taking this to the extreme, when we try to implement a statement like ``a system $S$ can be prepared in one of many states $\{\psi_k\}_k$'', whoever chooses the state $k$ does it  by measuring another system, obtain an outcome $k$, and then apply a physical transformation on their own memory and $S$ that prepares the state. We will see that in Spekkens' toy theory, physical agents are dramatically restricted in this action. Indeed, only agents outside the toy theory can prepare non-orthogonal states. In other words, toys can't play.

\paragraph{Quantum conditional preparation scenarios.} First consider a simple quantum  measure-and-prepare scenario: Alice measures a state $\ket{\phi}_R = \sum_k \alpha_k \ket{a_k}_R$  and, depending on her observed outcome $k$, prepares another system $S$  in state $\ket{\eta_k}_S$. 
This procedure can be described from the outside as a global two-step unitary process (Figure~\ref{fig:preparation}). First,  Alice's measurement is modelled by a unitary $V$ that couples $R$ to Alice's memory $A$; This is followed  by another unitary $U$, which implements the conditional state preparation, correlating $S$ with $A$. In the case of a strong measurement ($\inprod{\psi_k}{\psi_\ell}_A = \delta_{k\ell}$),  the overall transformation  acts as
\begin{align*}
    \left(\sum_k \alpha_k \ket {a_k}_R\right) \otimes \ket0_A \otimes \ket0_S \  \underrightarrow{\quad V \quad } \ \left(\sum_k \alpha_k \ket {a_k}_R \otimes \ket {\psi_k}_A \right)\otimes \ket0_S \ \underrightarrow{\quad U \quad }\  \sum_k \alpha_k  \ket {a_k}_R \otimes \ket{\psi_k}_A \otimes \ket{\eta_k}_S. 
\end{align*} 
After a strong quantum measurement ($\inprod{\psi_k}{\psi_\ell}_A = \delta_{k\ell}$), Alice is not restricted in the states $\{\ket{\eta_k}_S\}_S$ that she can conditionally prepare. In particular, she can conditionally prepare non-orthogonal states: for example, she can prepare $\ket0_S$ if she observes outcome $0$, and $\ket+_S$ if the outcome is 1,
\begin{align*}
    (\alpha \ket0_R + \beta \ket1_R)  \ket0_A  \ket0_S \  \underrightarrow{\quad V \quad } \ (\alpha \ket0_R \ket0_A  + \beta \ket1_R \ket1_A ) \ket 0_S \  \underrightarrow{\quad U \quad } \ \alpha \ket0_R \ket0_A \ket0_S + \beta \ket1_R \ket1_A \ket +_S.
\end{align*}
In contrast, the toy theory imposes unexpected constraints, and this kind of conditional preparation is not allowed.

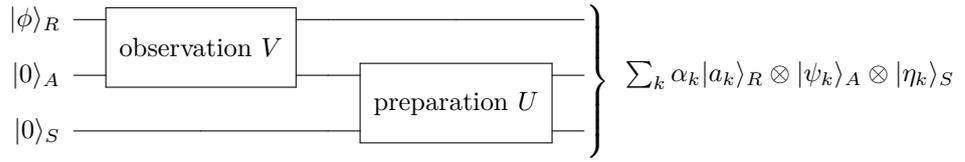
\begin{figure}[t]
\centering
\begin{align*}
        \Qcircuit @C=1.2em @R=1.2em {
        \lstick{\ket{\phi}_R} & \multigate{1}{\text{observation } V} & \qw & \qw & \qw \\
    \lstick{\ket{0}_A} & \ghost{\text{observation } V} & \qw & \multigate{1}{\text{preparation } U} & \qw & \hspace{13em} \sum_k \alpha_k  \ket {a_k}_R \otimes \ket{\psi_k}_A \otimes \ket{\eta_k}_S \\
        \lstick{\ket{0}_S} & \qw & \qw & \ghost{\text{preparation } U} & \qw
        \gategroup{1}{5}{3}{4}{1em}{\}} \\
        }
        \end{align*}
\caption{{\bf Conditional preparation.}  Alice first observes the system $R$, writing down the outcome to her memory, and then initializes the state of the system $S$ depending on her outcome.}
\label{fig:preparation}
\end{figure}

\paragraph{Forbidden conditional preparations in the toy theory.} Conditional preparations of non-orthogonal states are not allowed in the {toy} theory. 
The problem is not in Alice's measurement (which we've seen are very similar to quantum measurements), but in the second step, $U$, where Alice performs the conditional preparation of non-orthogonal states. 
This transformation would have to act on the joint state of $A$ and $S$ analogously to 
\begin{align*}
     U: \quad \ket0_A \ket0_S &\to  \ket0_A \ket0_S, \\ 
    \ket1_A \ket0_S &\to  \ket1_A \ket+_S, 
\end{align*} 
which in the toy theory would look like 
\begin{align*}
\toys{00}  \xrightarrow{U}  \toys{00}, \quad \toys{10}  \xrightarrow{U}  \toys{1+}
\end{align*}
There are no allowed transformations that implement this action in the toy theory, even when we consider transformations on a larger system, which are irreversible at this scale (\changemarker{Corollary}~\ref{lemma:no_conditional_transf_example}). 
The generalization of this example is a dramatic restriction on agents' actions;  see Theorem~\ref{thm:condtrafo} for the formal version of this  result. \changemarker{This theorem applies to arbitrary system dimension and arbitrary amount of systems, that can later be traced out.}

\begin{theorem}[Restrictions on conditional action of agents in the Spekkens' toy theory]\label{thm:conditional_action_informal}
In Spekkens toy theory, if an agent measures a system $R$, obtaining outcome $k$, and prepares a second system $S$ in one of several states $\{\psi_k\}_k$  depending on the outcome $k$, then any two of these states $(\psi_k, \psi_\ell)$ must be either identical or orthogonal. The number of identical states of each type is the same.
\end{theorem}

\section{Forgetting: valid expressions of ignorance}
\label{sec:forgetting}
\paragraph{Forgetting with explicit quantum memories.}
{The physical process of forgetting information originally stored in a memory can be modelled through an interaction between the memory and its environment; this may lead to the loss of correlation between the current memory content and the system it referred to.  For example, suppose that you write your credit card number on a notepad --- the notepad is correlated with your credit card. If later the notepad is smudged, erased or burned (all physical interactions with an environment), it will no longer be perfectly correlated with the credit card. The same can be said of information stored in a (quantum or classical) hard drive which is subject to noise and decoherence originating from the interaction with its environment, as expressed by the data processing inequality. In a simple example (Figure~\ref{fig:forgetting}), suppose that the agent measures a bipartite system $S = S_1 \otimes S_2$, storing the outcome in their bipartite memory $M=M_1 \otimes M_2$, through a standard von Neumann scheme which from an outside perspective is modelled like a coherent copy operation, entangling $S$ and $M$,}
\begin{align*}
    \left(\sum_{k,j} \alpha_{kj}\ \ket{\phi_k}_{S_1} \ket{\psi_j}_{S_2}  \right) \otimes \ket{0}_{M_1} \ket{0}_{M_2} 
    &\longrightarrow \sum_{k,j} \alpha_{kj}  \ket{\phi_k}_{S_1} \ket{\psi_j}_{S_2} \otimes \ket{k}_{M_1} \ket{j}_{M_2} =: \ket{\Psi}_{SM}.
\end{align*}
{Now let the memory interact with an  environment system $E$. An example is complete thermalization of the second register, in which the environment is initially in a thermal state $\tau_E$ and the interaction swaps the state of the second register  $M_2$ with the environment,}
\begin{align*}
    U_{ME}:  \id_{M_1} \otimes \operatorname{SWAP}_{M_2 E}
\end{align*}
{After the global state evolves under $\id_S \otimes U_{ME}$, the final (mixed) state of $S$ and the memory is }
\begin{align*}
    \rho_{SM} &= \tr_E([\id_S \otimes U_{ME}]\  [\proj{\Psi}{\Psi}_{SM} \otimes \tau_E ]\ [\id_S \otimes U^\dagger_{ME}])\\
    &=
    \sum_{k, k', j} \alpha_{kj}\ \alpha_{k'j}^* \ \proj{\phi_k}{\phi_{k'}}_{S_1} \otimes \proj{\psi_j}{\psi_j}_{S_2} \otimes \proj{k}{k'}_{M_1} \otimes \tau_{M_2},
\end{align*}
{where all the information about $S_2$ is lost to the environment, as can be seen from the mutual information}
\begin{align*}
     I(S_2:M)_\rho &= 0,  \quad I(S_1:M)_\rho = I(S_1:M)_\Psi. 
\end{align*}

\paragraph{{Abstract uncertainty in quantum theory.}}
{In addition to this physical process of forgetting, in quantum theory mixed states $\rho_S = \sum_i p_i \rho_i$ can be used to describe a  state of knowledge  where an agent has abstract uncertainty about which of the states $\{\rho_i\}_i$  describes $S$, and their best Bayesian guess is described by a probability distribution $\{p_i\}_i$.  In quantum theory the probabilities and states in these abstract mixtures can be arbitrary, and one can always find a physical forgetting process on an explicit memory that connects the physical and abstract representations. 
In  the toy theory (as it stands at the time of writing) descriptions of uncertainty are severely limited, perhaps because there are no natural physical sources for this uncertainty in the toy world.}

\begin{figure}[t]
    \centering
    \begin{align*}
        \Qcircuit @C=1.4em @R=1.4em {
        & &   \mbox{\qquad Memory update} & & \\
        \lstick{R_1} & \qw & \ctrl{2} & \qw & \qw & \qw & \qw  \\
        \lstick{R_2} & \qw & \qw & \ctrl{2} & \qw &\qw & \qw \\
        \lstick{M_1} & \qw & \targ & \qw & \qw & \qw & \qw \\
        \lstick{M_2} & \qw & \qw & \targ & \qw & \qswap & \qw \\
        \lstick{E} & \qw & \qw & \qw & \qw & \qswap \qwx & \qw & \rstick{\text{Interaction memory/environment}} \gategroup{2}{3}{5}{4}{1.1em}{--} \gategroup{5}{5}{6}{7}{1.1em}{--} \\
        }
        \end{align*}
    \caption{{\bf Losing information as a quantum evolution.} One can imagine a process where the information stored in an agent's memory is partially lost to the environment. Here, the agent first performs a memory update, writing down the outcomes of measurements on systems $S_1$ and $S_@$ in their memory systems $M_1$ and $M_2$. Due to an interaction with the environment $E$ (for example a complete thermalization represented here as a SWAP gate), the information contained in memory qubit $M_2$ is exchanged with the environment, so that correlations between the memory and $S_2$ are lost to the agent through this process.}
    \label{fig:forgetting}
\end{figure}
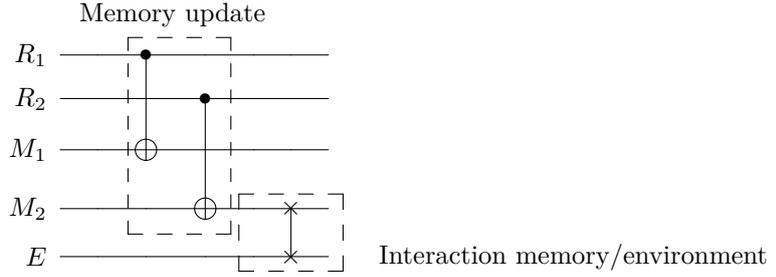

\paragraph{Forgetting with an implicit toy memory.} In the toy theory,  not all states can be mixed in a way such that the resulting state is allowed by the epirestricted picture. For example, while it is possible to mix the states
\begin{align*}
    \toys{00} \bigvee \toys{10}  = \toys{mixedleft},
\end{align*}
we are not allowed to mix states as 
\begin{align*}
    \toys{00} \bigvee \toys{10}  \bigvee \toys{11} 
    = \twocustomtoys{rrrr}{rrrr}{rr00}{rr00} !
\end{align*}
This means that for certain sets of states we are not allowed to forget which states we had initially. Moreover, even if we are able to 'forget', we only forget each state with an equal probability: the epistemic states always constitute uniform probability distributions over the corresponding set of ontic states.  
One can argue that only uniform distributions over the states we choose to forget (when we can) are physical, as it is not clear how one would assign non-uniform priors to the probabilities of forgetting for different states. For example, we cannot model a setting where we forget that the system is in the state $\{1,2\}$ with probability $\frac{1}{3}$, and in the state $\{3, 4\}$ with probability $\frac{2}{3}$ --- we are only allowed to forget both states with an equal probability of $\frac{1}{2}$.

\paragraph{{Forgetting as a physical process with explicit toy memories.}}{ We can now see the equivalent of the quantum example where we explored the physical process of forgetting. Suppose that $S$ starts in state }
\begin{align*}
   \begin{tabular}{c c }
         $S_1$ & \toys{1+} 
        \vspace{2mm}  \\ 
         & $S_2$
    \end{tabular}   = {\toy1}_{S_1}\otimes {\toy+}_{S_2}.
\end{align*}
{and the agent has a memory $M=M_1\otimes M_2$ originally in state }
\begin{align*}
   \begin{tabular}{c c }
         $M_1$ & \toys{00} 
        \vspace{2mm}  \\ 
         & $M_2$
    \end{tabular}   = {\toy0}_{M_1}\otimes {\toy0}_{M_2}.
\end{align*}
{The agent measures the two toy bits of $S$ in the $Z$ basis, storing the outcome of the first measurement in  $M_1$ and the second in $M_2$. From the outside we can model this measurement as an  entangling operation $\operatorname{CNOT}_{S_1 M_1}\otimes \operatorname{CNOT}_{S_2M_2}$ resulting in the global state}
\begin{align*}
    \begin{tabular}{c c }
         $S_1$ & \toys{11} 
        \vspace{2mm}  \\ 
         & $M_1$
    \end{tabular}  
    \otimes  \
   \begin{tabular}{c c c c}
         $S_2$ & \toys{bell}
         \vspace{2mm}  \\ 
         & $M_2$
    \end{tabular}  .
\end{align*}
{Now we simulate the decoherence process in memory $M_2$, through swapping with an environment in a fully mixed state, $\scalebox{0.7}{\toy{mix}}_E$. In the quantum analogy, this corresponds to full thermalization of $M_2$ with an environment  at infinite temperature or with a degenerate Hamiltonian. The grid diagram of the joint state of $S_2 \otimes M_2 \otimes E$ would span three dimensions, but a convenient visualization is} 
\begin{align*}
    \begin{tabular}{c c c c}
         $S_2$ & \toys{bell}
         \vspace{2mm}  \\ 
         & $E$
    \end{tabular} 
    \otimes {\toy{mix}}_{M_2},
\end{align*}
{from which it is easier to see that the reduced state of $S_2 \otimes M_2$  is }
\begin{align*}
    {\toy{mix}}_{S_2} \otimes {\toy{mix}}_{M_2}
    = \ 
    \begin{tabular}{c c c c}
         $S_2$ & \toys{mix}
         \vspace{2mm}  \\ 
         & $M_2$
    \end{tabular} .
\end{align*}
{This gives us the joint state of $S$ and $M$}  
\begin{align*}
    \begin{tabular}{c c }
         $S_1$ & \toys{11} 
        \vspace{2mm}  \\ 
         & $M_1$
    \end{tabular}  
    \otimes  \
   \begin{tabular}{c c c c}
         $S_2$ & \toys{mix}
         \vspace{2mm}  \\ 
         & $M_2$
    \end{tabular}  .
\end{align*}
{In this example, measuring their memory $M$ could give the agent some information about $S_1$ but not about $S_2$ --- they have forgotten about $S_2$.  This physical picture  helps us understand why forgetting in arbitrary ways isn't allowed in the toy theory: to model partial ways of forgetting, we can vary the interaction $V_{ME}$ between memory and the environment, and the initial state of the environment. Because both transformations and epistemic states are restricted in the toy theory, the final states of $S\otimes M$ are also restricted in form. The information an agent still remembers about $S$ after interaction with an environment can be again modelled through the reduced states of $S$  conditioned on a measurement on the memory, and these states are, as we have seen, restricted. }

\paragraph{Interpretation.}
{Another intuition for why forgetting always results in a fully mixed state lies in how we understand knowledge in the toy theory~\cite{Lostaglio2021}. The principle of knowledge balance imposes that we either possess information about the individual states of the systems, or information about correlations between them. After a physical measurement, the memory and the measured system are perfectly correlated, so from the outside we have no information about the reduced state of the system; erasing the memory effectively erases the information about correlations, leaving us maximally ignorant about the state of the measured system.}

\section{Reasoning: making inferences about other agents' experiments}
\label{sec:reasoning} 
In this section, we formalize the conditions under which agents can reason about measurement outcomes -- their own and each other's. First, we look at what it means to get a certain outcome or predict it with certainty. Then, we apply this result to a particular subset of statements agents can make, namely, inferential statements. Finally, we demonstrate how these rules are applied, using examples of Bell scenario, Wigner's friend, and Frauchiger-Renner thought experiment in the toy theory, and discuss the differences from their quantum counterparts.

\paragraph{Deterministic predictions.} In the following thought experiments, agents are able to reason about each other's outcomes --- for deterministic statements. Let us formalize what it means to measure something with certainty or to predict something with certainty in the toy theory. 
In Appendix~\ref{appendix:predproofs} we prove \cref{para:ingr:cond} which gives two conditions for an outcome of a mesurement to happen with certainty. The first condition certifies that an outcome happens with certainty, while the second condition ensures that this outcome is the desired one.

\paragraph{Conditions for making inferences.}
In thought experiments we often make inferences of the type ``A = 1 $\implies$ B = 1'', which predict other agent's measurement outcome based on our observation. How do we formalize the certainty of such an inference in the toy theory? The statement ``A = 1 $\implies$ B = 1'' corresponds to the conditional probability $P(B = 1| A = 1) = 1$. 
\Cref{pred_cert} gives two conditions for when we can make valid inferences. 
Intuitively, the first condition ensures that there is an outcome of the measurement of observable $B$ that can be inferred if observable $A = 1$ is known. The second condition ensures that the outcome that can be inferred is indeed the outcome $B = 1$.

\subsection{Example: a Bell scenario}\label{sec:bell_scenario}

\begin{figure}[t]
\centering
       \includegraphics[scale=0.25]{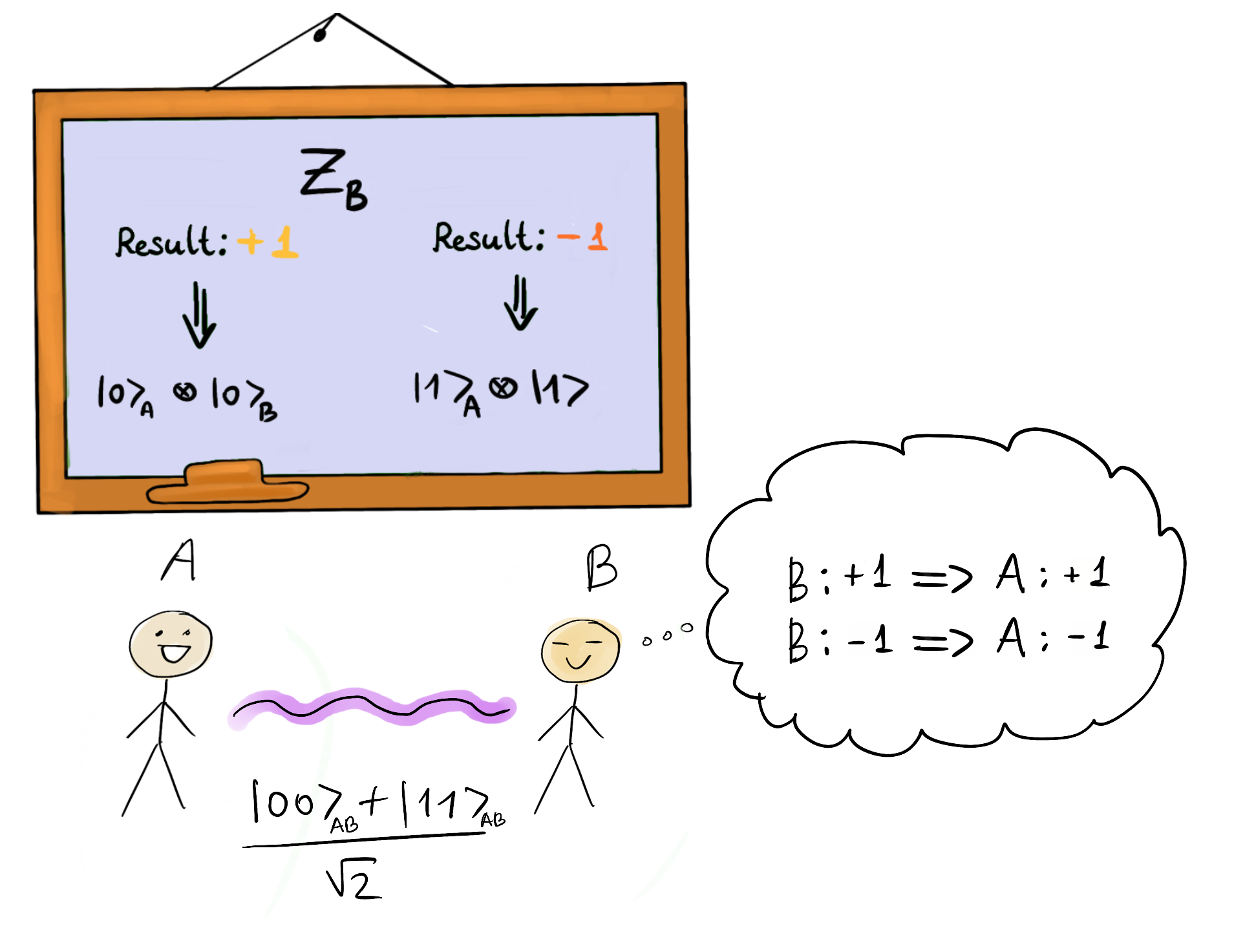}
\caption{{\bf Reasoning in a quantum Bell experiment.} Alice and Bob share a Bell state, and each perform a $Z$ basis measurement on their qubit.  Given his outcome, Bob reasons about Alice's outcome and can predict it with certainty.}
  \label{fig:bell-qm}
\end{figure}

\paragraph{Quantum Bell setting.} In quantum theory, if Alice and Bob share a Bell state and measure their individual qubits in the computational basis (corresponding to the observables $Z_A$ and $Z_B$), they can make inferences about each other's outcomes (Figure~\ref{fig:bell-toy}). For example, if the shared state is
\begin{gather*}
    \frac{1}{\sqrt{2}} \left( \ket{00}_{AB} + \ket{11}_{AB} \right),
\end{gather*}
and Bob obtains outcome $b=0$, he can  update his description of the global post-measurement state to $\ket{00}_{AB}$ and infer  with certainty that Alice must obtain outcome $a=0$; analogously for $b=1$.

\begin{figure}
\centering
        \includegraphics[scale=0.25]{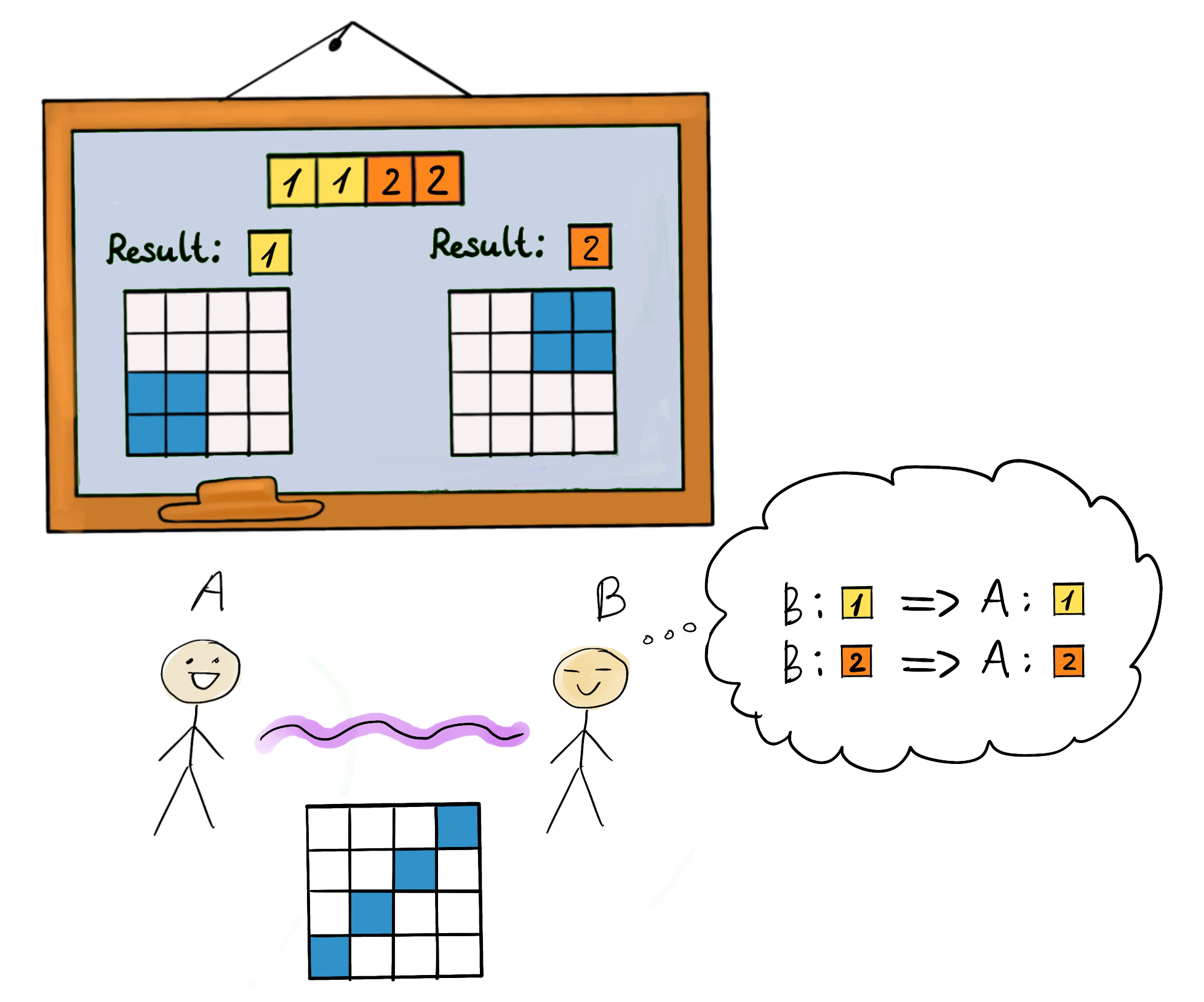}
        \caption{{\bf Reasoning in a toy Bell experiment.}
        Analogously to the quantum Bell scenario, Alice and Bob share an entangled bipartite system and perform the toy version of local $Z$ measurements. Based on his outcome, Bob can predict Alice's outcome with certainty.}
        \label{fig:bell-toy}
\end{figure}

\paragraph{Toy Bell scenario.} This scenario can be reenacted in the toy theory, with similar results: Bob can make a deterministic prediction about Alice's outcome (Figure~\ref{fig:bell-toy}). Alice and Bob share the entangled state analogous to a Bell state,
\begin{align*}
\toys{bell}
\end{align*}
Bob measures his system in the  toy-$Z$ basis,  $\mathcal M_Z = \toymeasurementz$. If he obtains outcome $B= 0 = \lonecell{orange}$, he can  update his description of the shared state  to 
\begin{align*}
    \toys{00}. 
\end{align*}
From this description, Bob can infer that if Alice now measures her system in the same basis, she will obtain outcome $A= 0 = \lonecell{orange}$  with certainty, that is, ``$B = 0 \implies A = 0$''  Analogously, he can conclude that ``$B = 1 \implies A = 1$''.  A formal proof can be found in Appendix~\ref{appendix:predproofs}.

\subsection{Example of meta measurements: Wigner's friend}

\paragraph{Wigner's quantum friend.} Wigner's friend experiment was first proposed by Wigner~\cite{Wigner1961}. The setting involves a quantum system $R$ and an observer $A$ (Alice) performing a measurement on this system in a closed laboratory, as well as an outside observer Wigner. For Alice in the lab, the outcome of the experiment is recorded in the device she is using to measure the system $R$, for example, as a position of a pointer (or an entry in her memory). However, Wigner does not have any information about Alice getting a particular outcome, and describes the evolution of the closed lab as a unitary (reversible) process, and assigns an entangled state to $R$ and $A$. In the language of quantum mechanics, if we assume that the pointer is initially in the state $\ket{0}_A$, and the state of the measured system is $\frac{1}{\sqrt{2}}\ket{0}_R+\frac{1}{\sqrt{2}}\ket{1}_R$ this process corresponds to
\begin{equation}
    \left(\frac{1}{\sqrt{2}}\ket{0}_R+\frac{1}{\sqrt{2}}\ket{1}_R\right)\ket{0}_A \rightarrow \frac{1}{\sqrt{2}}\ket{0}_R\ket{0}_A+\frac{1}{\sqrt{2}}\ket{1}_R\ket{1}_A
\end{equation}
Alice and Wigner turn out to have descriptions of the same setting which are vastly different from each other. We will not discuss numerous conceptual implications of the original thought experiment here -- a review can be found in~\cite{Nurgalieva2021}. However, we would still like to see how we can model this setting in the toy theory (here we will do so in the original epirestricted picture).

\paragraph{Wigner's toy friend.}
In the toy theory, we consider again two subsystems: the measured system $R$ and Alice's memory register $A$. The individual states and the joint state of the systems $R$ and $A$ can be pictured as  
\begin{equation*}
  \begin{tabular}{c c c c}
       $R$ \toy+, &\quad  $A$ \toy0, &\quad $R$ &\toys{+0}  \\ \vspace{2mm} \\
       & & & $A$
  \end{tabular}
\end{equation*}
Alice measures $R$ in the toy-$Z$ basis. She describes her measurement as  $M_Z = \scalebox{0.7}{\toymeasurementz}$ and sees a definite outcome $a$.  Wigner sees Alice's measurement as a reversible transformation (the toy CNOT), and updates his description of the joint state of $R$ and $A$ as 
\begin{equation*}
    \begin{tabular}{c c c c}
         $R$ & \toys{+0} & $\xrightarrow{CNOT}$  & \toys{bell}
          \\ \vspace{2mm}  \\ 
         & $A$
    \end{tabular}  
\end{equation*} 
Alice, on the other hand, has the subjective experience of seeing one outcome $a$ and write it to her memory. We can describe the knowledge update from the perspective of the different agents as 
\begin{equation*}
        \begin{tabular}{c c c c m{6cm} }
         $R$ & \toys{+0} & $\xrightarrow{a = 0}$  & \toys{00}  & Alice, seeing outcome $a=0 = \lonecell{orange}$\\ \vspace{2mm}  \\ 
         & $A$ & or \\ 
         & & $\xrightarrow{a = 1}$  & \toys{11}  & Alice, seeing outcome $a=1 = \lonecell{g}$\\
         \vspace{6mm}  \\  & & or \\
         & &  $\xrightarrow{CNOT}$  & \toys{bell} & Wigner, describing Alice's measurement as a reversible CNOT transformation.
    \end{tabular} .
\end{equation*}

\paragraph{Interpretation.} In the framework of the toy theory, the difference between Alice's and Wigner's descriptions has a straightforward interpretation. Thanks to the knowledge balance principle, in the state of the maximal knowledge an agent can have maximal information either about an individual system or  about how these individual systems are correlated. Hence, Alice's and Wigner's epistemic states do not contradict each other, and simply represent two different ways an agent can view a composite system (two states of knowledge about the same ontic state).
In~\cite{Lostaglio2021}, it is shown that the Wigner and Alice's views discrepancy can be reproduced (and interpreted!) in any realistic toy model. Alice's collapsed state can simply be understood as more coarse-grained compared to Wigner's; the correlations of her state are beyond her description level. She can still get away with it, though, as the correlations are not used in later dynamics --- which is not the case for the next scenario we present.

\subsection{Example of multi-agent paradoxes: the Frauchiger-Renner scenario}

\begin{figure}[t]
    \centering
    \includegraphics[scale=0.3]{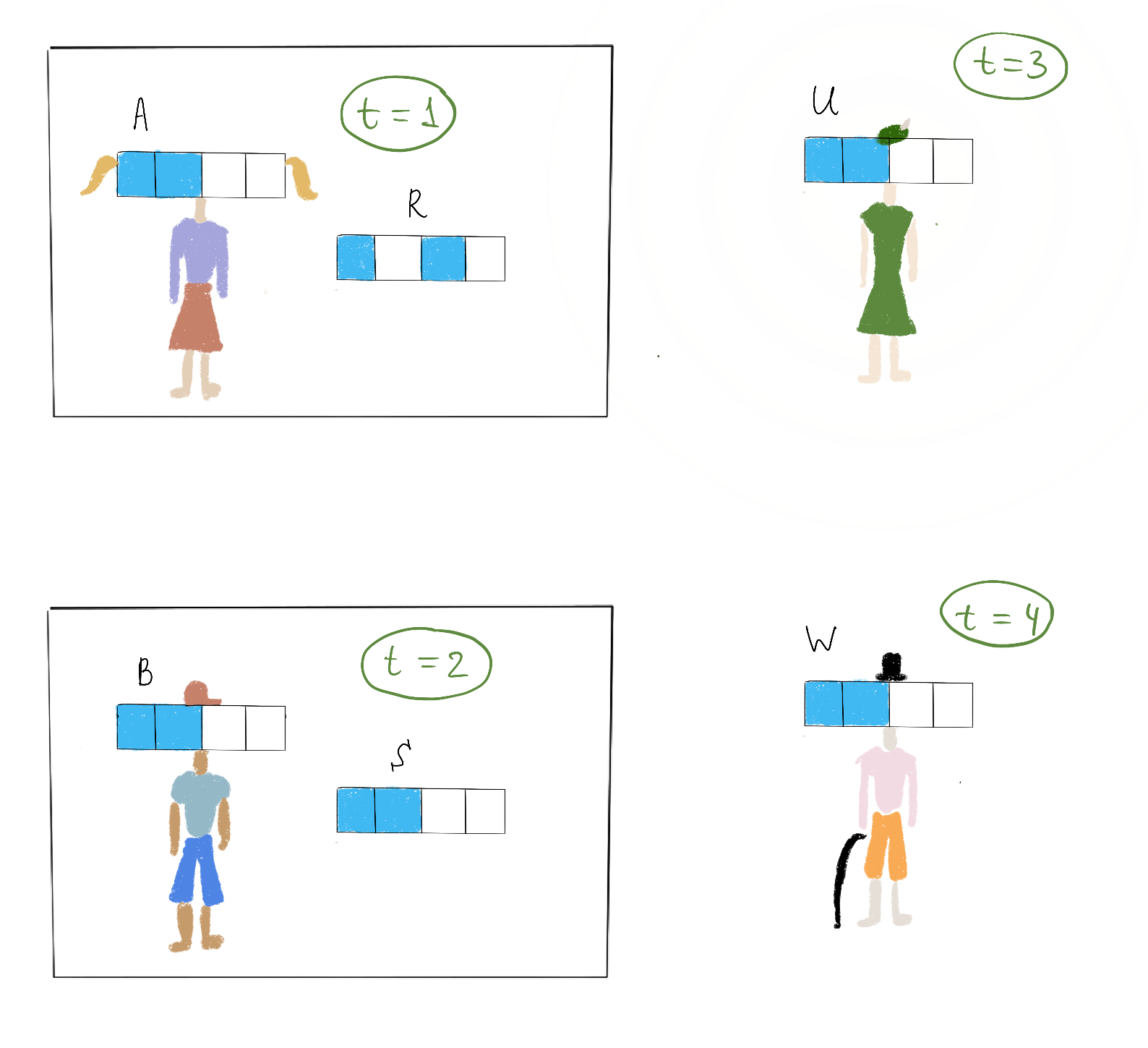}
    \caption{{\bf The Frauchiger-Renner-like setup of agents in the toy model.} The setting includes four agents $A$, $B$, $U$ and $W$, memories of which are all modeled as two-level systems in the toy model, and two additional two-level systems $R$ and $S$ in Alice's and Bob's labs respectively. In the toy version of the FR setup, we are not concerned with the course of the protocol four agents follow. The question we ask is whether there exists a joint state of systems $A$, $B$, $R$ and $S$, and measurements performed on the systems $R$, $S$, $RA$ and $SB$, such that the agents come to the same contradicting chain of statements as in the original quantum setting~\cite{Frauchiger_2018}. A summary of the quantum experiment can be found in Appendix~\ref{appendix:fr-quantum}.  } 
    \label{fig:fr}
\end{figure}

\paragraph{Quantum multi-agent paradox.}
In the Frauchiger-Renner setting~\cite{Frauchiger2018}, four quantum agents (Alice, Bob, Ursula and Wigner)  perform a series of measurements and reasoning steps, reaching a logical contradiction,
\begin{equation}
    (w = \text{ok} \wedge u = \text{ok}) \implies b = 1 \implies a = 1
\implies  w = \text{fail},
\end{equation}
that is, when both Ursula and Wigner obtain outcomes ``ok'' in their measurements, they can reason based on their observation that that Bob predicted that Alice predicted that Wigner would obtain a different outcome, ``fail'', with certainty.  The experimental protocol consists of individual steps that we covered so in this manuscript: two qubits
$R$ and $S$ are initially prepared in an entangled Hardy state \cite{Hardy1993}; Alice measures $R$  and Bob measures $S$; then Ursula measures Alice's lab (including $R$ and Alice's memory $A$), and finally Wigner measures Bob's lab (including $S$ and Bob's memory $B$). The contradiction is found for a specific choice of initial state and measurement bases, which are  described in Appendix~\ref{appendix:fr-quantum}. \footnote{In the original formulation of the experiment, Alice measures $R$ and based on the outcome she makes a conditional preparation of $S$ in one of two non-orthogonal states, which she sends to Bob. The two formulations (conditional preparation or initially entangled state) are equivalent  for the logical analysis of the experiment, as they express the same correlations.}
For a pedagogical discussion of the original paradox and the assumptions behind it, we refer to our previous work~\cite{NL2018}; for discussions of broader implications for abstract logic and  interpretations of quantum theory see for example~\cite{Fraser2020,Boge2019,Nurgalieva2021,Haddara2022}. The paradox has also been shown to arise in other physical theories, namely box world~\cite{Vilasini_2019}.

\paragraph{Toy multi-agent scenario.}
Our question is whether a similar multi-agent logical paradox can be found in Spekkens' toy theory; we will see that it cannot, partially because of the restrictions in the individual operations like conditional state preparation, and partly because this is an explicitly non-contextual epistemic theory.
Since the toy theory does not allow Alice to perform non-orthognal conditional state preparation of $S$, we follow the version of the experiment where all relevant correlations are encoded in the initial state of $RS$. 
The global system of Alice and Bob's labs is composed of four subsystems: two systems $R$ and $S$ measured by Alice and Bob, and Alice's and Bob's memory registries $A$ and $B$ (Figure~\ref{fig:fr}). The order of measurements follows the original experiment, and the systems can be of an arbitrary dimension $k$. 
We show that in the toy theory, there is no choice of initial state and measurements by the four agents that can lead to a logical contradiction in this experimental scenario. A formal description of the setting and proof can be found in Appendix~\ref{appendix:fr-proof}. 

\begin{restatable}{theorem}{parathm}\label{para:thm}
    There exists no valid epistemic state that can be used to
    reach a logical contradiction in the Frauchiger-Renner setting in the toy theory.
\end{restatable}

\section{Discussion}
\label{sec:conclusions}
\paragraph{Learning, reasoning and forgetting as physical processes in the toy theory.} In this work, we found a way to model the physical evolution of systems and agents'  memories that implement in Spekkens' toy theory the abstract process of measurement and forgetting information, analogously to how quantum measurements and information loss are modelled as explicit quantum evolutions. We found conditions on experimental settings that guarantee that agents can reason with certainty about each other's experiments, including settings where agents can measure each other.

\paragraph{Restrictions on free choice of agents.} One can interpret the impossibility of an arbitrary conditional preparation in the toy theory as a limitation on the free choice of agents. An experimenter cannot decide to prepare an arbitrary valid epistemic state depending on her observations --- she is constrained to an orthogonal set of states. In addition,  agents cannot set arbitrary probability distributions as inputs for future experiments (because such distributions would be encoded in physical systems like  biased coins).  Note that agents can still perform deterministic operations that entangle other systems; they just cannot make a decision about which entangling operation to apply that does not result in uniformly-distributed orthogonal states on those systems.

\paragraph{In the toy theory,  limited knowledge is\dots limited.}
While in classical and quantum theories we can always lose information, in Spekkens' toy theory it is impossible to model many natural expressions of limited knowledge, like not knowing which of two non-orthogonal states a system is in; conditional state preparation of arbitrary states is also forbidden. This shows us that even aspects of logic and information theory that we take for granted and consider independent of the physical theory (like having probabilistic knowledge) are indeed dramatically physical. 

\paragraph{Foils of the toy theory}
To understand to which extent these peculiarities are an artifact of the knowledge balance principle, one could consider possible relaxations of the principle\footnote{{``I think I know what to do. We're gonna have to break a few rules, but if it works, it'll help everybody.'' --- Toy story}} --- for example, imposing it for pure states but allowing all probabilistic mixtures of pure states. One must proceed with caution, as any such relaxation may have unintended consequences for the stability of the theory: 
for example, we considered the relaxation in the original formulation where we only require valid marginals, but not that validity is preserved under subsystem measurement, and tried to come up with a different measurement update that does not require this. However, we found that it is not possible to define such a measurement update and, additionally, found mixed epistemic states that that cannot be written as the mixture of pure states. We leave the investigation of other relaxations of the theory to future work.

\paragraph{Forgetting in other epistemic theories.}
Epistemic models provide insight into how different epistemic restrictions influence the set of transformations and measurements an agent can perform, and how their memory can be modeled. We have already mentioned epistemically restricted Liouville mechanics~\cite{Bartlett_2012}, and here we have taken a look at a particular epistemic limitation of ``knowledge balance principle''. However, one can imagine that the restrictions above can be weakened or modified. One possible direction of the future research then would be modeling memories of agents and their reasoning for an arbitrary relation between information contained in the epistemic and ontic states of the system, and see in which cases conclusions made in quantum mechanics are reproduced. This could lead to a better understanding of which types of epistemology can be admitted by quantum mechanics, and what are essential properties of such epistemic theories.
We leave the investigation of how the striking limitations in the process of forgetting information plays out for other (epistemically restricted) theories as future work.

\paragraph{No multi-agent logical paradox in the Frauchiger-Renner setting.}
We proved that in settings analogous to the Frauchiger-Renner experiment there is no assignment of states and measurements that can lead to a logical paradox in the toy theory. We don't claim that our model of agents is exhaustive; in our analysis, we only capture one degree of freedom of the agent which corresponds to the memory register for the outcome of their measurement, and for this particular model (albeit a minimally reasonable one) the paradox does not come to be.  We conjecture that there is no model that can lead to a paradox, in arbitrary multi-agent settings, because the theory is non-contextual; we will investigate this in future work.

\paragraph{Relation to contextuality.}
Multi-agent logical paradoxes  involve chains (or possibly more general structures) of statements that cannot be simultaneously true in a consistent manner. Failures of noncontextuality  can often be expressed in terms of the inability to consistently assign definite outcome values to a set of measurements \cite{Kochen1967,Spekkens_2005}. 
Examples of paradoxical chains of reasoning in quantum theory~\cite{Frauchiger2018} and box world~\cite{Vilasini_2019} --- two contextual theories --- and the intuition of the impossibility of finding such a chain in Spekkens' toy model, as shown here, suggests the following conjecture: logical multi-agent paradoxes are proofs of contextuality, and all contextual physical theories can model multi-agent logical paradoxes. 
The connection of contextuality to logical contradictions has already been to some extent explored in existing research. For example, it can be shown that the patterns of reasoning which are used in finding a contradiction in the Liar cycles~\cite{Cook2004} are similar to the reasoning we make use of in FR-type arguments, and in~\cite{Abramsky15} the connection is established between such logical cycles and contextuality. 
Additionally, in \cite{Pusey2015}, it has been shown that every proof of a logical pre-post selection paradox is a proof of contextuality. 
The question of how proofs of contextuality  relate to proofs of multi-agent logical paradoxes will be formally addressed in future work. 

\paragraph{Weak and noisy measurements.} 
In~\cite{Karanjai2015}, the analysis of weak measurements and weak values is applied to the epistemically-restricted theory of Liuoville classical mechanics. Weak values in that case coincide with the ones obtained for gaussian quantum mechanics; no anomalous weak values are observed, as the theory is non-contextual. It would also be interesting to apply our analysis of physical measurements to try to implement noisy and weak measurements in generalized Spekkens' theory; we leave this as an open project.

{\paragraph{Wigner's other friends.}
The thought experiment analyzed in this paper has many similarities to the thought experiments proposed by Brukner~\cite{Brukner2018} and Cavalcanti~\cite{Cavalcanti2020}, which also build on the original Wigner's friend scenario. However, the conclusions drawn from the latter two differ from the original FR experiment and its toy analogue discussed in this paper. While Brukner's and Cavalcanti's results provide a strengthening of the Bell's theorem, considering the FR thought experiment in various theories is an exploration of what it means to be a user of the theory and also be described within the said theory, and what operational restrictions such a user might have. As the toy theory does not exhibit non-local features, and its correlations do not violate Bell's inequalities, it is not suitable for formulating Brukner's and Cavalcanti's results.
We leave the interesting question of identifying fundamental connections between the assumptions used in all of these scenarios as future work.}

\paragraph{Open questions and generalizations.} 
We did not use the explicit form of the epistemic restriction in our proof of non-existence of the paradox in the toy theory; the assumptions we made are not unique to classical complementary.
In principle, we could have defined a different criterion for joint knowability of observables.
To preserve the general structure of the theory this new criterion needs to be subject to some conditions: for example, we require jointly knowable variables to have a linear structure; this follows from the fact that if two variables are known, any linear combination can be calculated. 
The formalization and exploration of different epistemic restrictions, as well as the investigation of other more general settings in which agents could reason about each other, is left for future work.

\subsection*{Acknowledgements} 
We thank Matthew F.~Pusey, {David Schmid, Y\`il\`e Y{\=\i}ng and Rob Spekkens} for pointed discussions and the anonymous NJP referees for valued feedback that led us to expand the pedagogical review in the introduction.
NN and LdR acknowledge support from the Swiss National Science Foundation through SNSF project No.\ $200020\_165843$ and through the National Centre of Competence in Research \emph{Quantum Science and Technology}(QSIT). 
LdR further acknowledges support from  the FQXi large grant \emph{Consciousness in the Physical World}. {LdR is grateful for the hospitality of Perimeter Institute where part of this work was carried out. Research at Perimeter Institute is supported in part by the Government of Canada through the Department of Innovation, Science and Economic Development and by the Province of Ontario through the Ministry of Colleges and Universities.}

\subsection*{Author contributions}
This manuscript is the second part of LH's semester project as a masters student (the first part was the review of the toy theory~\cite{Hausmann2021}). All authors contributed equally to the ideas and techniques developed here. LH wrote the first draft of the manuscript, and it was revised by all authors.


\newpage
\appendix


\section{Toy theory formalism in arbitrary dimensions}
\label{appendix:toyformalism}
The complete review of the toy theory formalism, including its original, stabilizer and arbitrary formulations, can be found in~\cite{Hausmann2021}.

\subsection{Formalism for arbitrary dimensions.}

\paragraph{Is this formalism necessary to understand the results of this paper?} To tackle the more general case of toy systems of arbitrary dimensions, we must review  heavier formalism \cite{Catani_2017, SpekkensFoundations2016}. Our main results are expressed in this language, but we add intuitive descriptions that convey the main message and don't require learning the formalism. 

\paragraph{Epistemic states.}
In the continuous case, we represent toy systems through observables $q_i$ and $p_i$ for each of the $n$ subsystems, analogous to position and momentum: for example, a toy particle moving in 3D would have $n=3$. For arbitrary dimensions, we represent a valid epistemic state $(\vec V, \vec v)$ by the known observables $\vec V = \langle \vec f_1, \dots, \vec f_k\rangle $ and the valuation $\vec v$ of those observables. For example,  if all we know about a 1D continuous system is that $\vec f_1=  2q_1 -  p_1 = 5$, then $\vec V = \langle \begin{pmatrix} 2 \vspace{1mm} \\ -1 \end{pmatrix}\rangle
$ and we can have for example $\vec v= \begin{pmatrix} 3 \vspace{1mm} \\ 1 \end{pmatrix} $, as $(2, -1) \begin{pmatrix} 3 \vspace{1mm} \\ 1 \end{pmatrix} = 2\times3 - 1\times1 = 5$. 
The set of ontic states compatible with this epistemic state are all  those that share the valuation $\vec v$ for the observables in  $\vec V$; this includes $o_1 = \begin{pmatrix} 3 \vspace{1mm} \\ 1 \end{pmatrix} $, but also for instance $o_2 = \begin{pmatrix} 2  \vspace{1mm} \\ -1  \end{pmatrix} $, as $(2, -1) o_2 = 5 $ , $o_3 = \begin{pmatrix} 4 \vspace{1mm} \\ 3  \end{pmatrix} $,   and so on.

\paragraph{Epistemic restriction for continous systems.}
The complete epistemic restriction in then given by the principle of classical complementarity:

        \begin{displayquote}
            ``The valid epistemic states are those wherein an agent
            knows the values of a set of quadrature variables that
            commute relative to the Poisson bracket, and is maximally
            ignorant otherwise.'' \cite{SpekkensFoundations2016}
        \end{displayquote}
        Here, ``maximal ignorance'' means that there is a uniform
        probability over all other values of variables.
        It can be shown that this complementarity principle requires observables
        to be linear, i.e.\ quadrature observables. 
        
        We represent quadrature observables as a vector $\vec{f} \in \mathbb{Z}^{2n}_d / \mathbb{R}^{2n}$ if we consider $n$ systems of dimension $d$ or $n$ continuous systems. Then the Poisson bracket is defined for both the
        continuous and discrete case, where the sum has to be understood
        in mod $d$ in the discrete case:
        \begin{align}
            [f,g] = \sum_{i = 1}^n f_{2i - 1} g_{2i} - f_{2i} g_{2i-1}
        \end{align}
        where $f_j$ denotes the $j$th entry of $\vec{f}$ \footnote{The poisson bracket can also be understood as the symplectic inner product of $\vec f$ and $\vec g$.} \footnote{Because of this definition we call the $2i-1$th entry the position of the system $i$ and the $2i$th entry the momentum of system $i$} \cite{SpekkensFoundations2016}.

\paragraph{Composing continuous systems.} To continue our example suppose that we bring in a second 1D system, and we know the local observable corresponding to the position of this system, for instance $\vec f_2 = q_2 = 10$. The global epistemic state is then specified by
$$
(\vec V, \vec v') = (\langle \vec f_1, \vec f_2\rangle, \vec v' )
=
\left(\left\langle 
\begin{pmatrix} 2 \vspace{2mm} \\ -1 \vspace{2mm} \\ 0 \vspace{2mm} \\ 0  \end{pmatrix} ,
\begin{pmatrix} 0 \vspace{2mm} \\ 0 \vspace{2mm} \\ 1 \vspace{2mm} \\ 0 \end{pmatrix}
\right\rangle, 
\begin{pmatrix} 3  \vspace{2mm} \\ 1 \vspace{2mm}  \\ 10 \vspace{2mm}  \\ 0 \end{pmatrix} \right),
$$
which is a product state.
On the other hand, if instead of $q_2$ we knew a global property, like that the positions of the two systems were perfectly correlated, $q_1=q_2$, we could represent this through a new observable $\vec f_3 = q_1 - q_2=0$, and so our global epistemic state would be 
$$
(\vec V', \vec v'') = (\langle \vec f_3\rangle, \vec v'' )
=
\left(\left\langle 
\begin{pmatrix} 1 \vspace{2mm} \\ 0 \vspace{2mm} \\ -1 \vspace{2mm} \\ 0 \end{pmatrix}
\right\rangle, 
\begin{pmatrix} 3  \vspace{2mm} \\ 1 \vspace{2mm}  \\ 3 \vspace{2mm}  \\ 0 \end{pmatrix} \right).
$$

\paragraph{Reversible transformations.}
Valid reversible transformations are \emph{sympletic transformations}, represented by a pair $(U, \vec a)$, where $\vec a$ is an ontic state\footnote{Throughout this paper $\vec a =\vec 0$ unless otherwise stated.}  and $U$ is a sympletic matrix. Sympletic matrices are those that  satisfy  $U^T J U = J$, where 
        \begin{equation}
            \vec{J} = \begin{pmatrix}
                0 & 1 & 0 & 0 & \dots \\[6pt]
                -1 & 0 & 0 & 0 & \\[6pt]
                0 & 0 & 0 & -1 & \\[6pt]
                0 & 0 & 1 & 0 & \\
                \vdots &  &  & &\ddots  \\
            \end{pmatrix}, 
        \end{equation}
and is used to write the sympletic inner product~\footnote{For two observables $\vec f$ and $\vec g$, their symplectic inner product corresponds to the Poisson bracket $[f,g] = \sum_{i = 1}^n f_{2i - 1} g_{2i} - f_{2i} g_{2i-1}$, with $f_j$ denoting the $j$th entry of $\vec{f}$. It can then be rewritten making use of the matrix $\vec J$: $[f,g] = \vec{f}^T\ \vec{J}\ \vec{g}$.}. The transformation $(U, \vec a)$ transforms an each ontic state from $\vec o$ to $U (\vec o + \vec a)$.
An epistemic state $(V,\vec{v})$ transforms under such a symplectic
transformation as
            \begin{equation}
                (V,\vec{v}) \to (\ (U^T)^{-1} V,\  U(\vec{v}+\vec{a})\ ).
            \end{equation}
The reason the transformation implements $(U^T)^{-1}$ on the vector space of known variables is that the transformation transforms an ontic state $\vec o$ to $U \vec o$, so if we know that an ontic state was compatible with the known variables before it must also be compatible afterwards.

For example consider the state 
$$\left(\vec V = \langle \begin{pmatrix} 2 \vspace{1mm} \\ -1 \end{pmatrix}\rangle, \vec v= \begin{pmatrix} 3 \vspace{1mm} \\ 1 \end{pmatrix} \right)
$$
and the transformation that swaps position and momentum
\begin{equation*}
    U = \begin{pmatrix}
                0 & 1 \\[6pt]
                1 & 0 
                
            \end{pmatrix}. 
\end{equation*}
This transforms the above state to
 $$\left(\vec V' = \langle \begin{pmatrix} -1 \vspace{1mm} \\ 2 \end{pmatrix}\rangle, \vec v'= \begin{pmatrix} 1 \vspace{1mm} \\ 3 \end{pmatrix} \right)
$$
The ontic state $o_2 = \begin{pmatrix} 2  \vspace{1mm} \\ -1  \end{pmatrix} $ which was compatible with the knowledge before the transformation is still valid after the transformation as 
$\begin{pmatrix} -1 \vspace{1mm} \\ 2 \end{pmatrix}^T (U o_2) = 5$.

\paragraph{Measurements.}
In the continuous case measurement consists of a vector space $V_{\pi}$ of observables that which can have outcomes 
$\vec{v}_{\pi} \in \Omega$, all possible (inequivalent) outcomes also result in a partition of the ontic state space, like in the discrete case.
The probability to get outcome $\vec{v}_{\pi}$ if the system
is in an ontic state $\vec{m}$ is given by \cite{SpekkensFoundations2016}
\begin{equation}
    \xi(\vec{v}_{\pi}|\vec{m}) = \delta_{V_{\pi}^{\perp} + \vec{v}_{\pi}}(\vec{m}).
\end{equation}
Intuitively, this means that we can only obtain measurement outcomes that are compatible with the ontic state of the system. 
We can denote a measurement $V_{\pi}$ and its outcome
$\vec{v}_{\pi}$ by the pair $(V_{\pi}, \vec{v}_{\pi})$. With
the conditional probability distribution $\xi(\vec{v}_{\pi}|\vec{m})$ we can calculate the
probability for a measurement outcome given the epistemic
state $(V,\vec{v})$  \cite{SpekkensFoundations2016},
\begin{equation}
    P(\vec{v}_{\pi}|(V,\vec{v})) = \sum_{\vec{m} \in \Omega} \xi(\vec{v}_{\pi}|\vec{m}) \  \mu_{(V,\vec{v})}(\vec{m}).
\end{equation}
where $ \mu_{(V,\vec{v})}(\vec{m})$ is the equal probability distribution over all ontic states compatible with the knowledge of the observables in $V$ having valuations $\vec v$.

\paragraph{Example of a measurement.}
For example consider the state 
$$\left(\vec V = \langle \begin{pmatrix} 2 \vspace{1mm} \\ 0 \end{pmatrix}\rangle, \vec v= \begin{pmatrix} 3 \vspace{1mm} \\ 1 \end{pmatrix} \right)
$$
and we want to measure the position
$$\langle \begin{pmatrix} 1 \vspace{1mm} \\ 0 \end{pmatrix}\rangle. $$
The epistemic state is such that we know the position to be $2 \cdot 3 = 6$, so the only compatible measurement outcome is 
$$\vec{v}_{\pi} = \begin{pmatrix} 6 \vspace{1mm} \\ 0 \end{pmatrix}.$$

\paragraph{Post-measurement state.} The post-measurement state of the system (given the information about the outcome) is as follows.

\begin{theorem}[Measurement update rule  ~\cite{Hausmann2021}]
\label{toy:gen:measurementupdate}
When an epistemic state $(V,\vec{v})$ is subjected to a measurement $V_{\pi}$, and outcome $\vec{v}_{\pi}$ is obtained, the epistemic state is updated to $(V', \vec v')$, where
         \begin{align}
                V' &= V_{\pi} \oplus V_{\text{commute}},\\
                \vec{v}' &\in (V_{\pi}^{\perp} +  \vec{v}_{\pi}) \cap (V^{\perp}_{\text{commute}} +
                \vec{v}),\\
            V_{\text{commute}}  
            &= \{\vec f \in V: \ [\vec f, \vec f_\pi] = 0, \ \forall \ \vec f_\pi \in V_\pi  \} 
            \subseteq V.
        \end{align} 
\end{theorem}




\section{Formal results and proofs}
\subsection{Linear algebra lemmas}
\label{appendix:la}
Here we list the results from linear algebra we used in the refinement of the generalization of Spekkens' toy theory.

  \begin{lemma}[\cite{Gallier_2011}]\label{toy:gen:lin1}
            Let $W \subset \Omega$ be a subvector space or
            submodule and $\vec{w} \in \Omega$. Then for any $\vec{a} \in W+\vec{w}$

            \begin{equation}
                W + \vec{w} = W + \vec{a}
            \end{equation}
        \end{lemma}

        \begin{proof}
            Let $\vec{b} \in W +\vec{w}$ then
            \begin{equation}
                \vec{b} = \vec{w}_1 + \vec{w} 
            \end{equation}
            for some $\vec{w}_1 \in W$. As $\vec{a} \in W + \vec{w}$ we know that 
            \begin{align}
                \begin{split}
                    \vec{a} &= \vec{w_2} + \vec{w} \\
                    \iff \vec{w} &= \vec{a} - \vec{w}_2 
                \end{split}
            \end{align}
            for some $\vec{w}_2 \in W$. Plugging the expression for
            $\vec{w}$ into the expression for $\vec{b}$ we find:
            \begin{equation}
                \vec{b} = \vec{w}_1 - \vec{w}_2 + \vec{a} \in W + \vec{a}
            \end{equation}
            Therefore, $W +\vec{w} \subset W + \vec{a}$.

            Let $\vec{c} \in W + \vec{a}$ then we can write 
            \begin{equation}
                \vec{c} = \vec{w}_3 + \vec{a} = \vec{w}_3 + \vec{w}_2 +\vec{w} \in W + \vec{w}
            \end{equation}
            where we used the expression for $\vec{a}$ from above. From this
            equation we can conclude that $W +\vec{a} \subset W + \vec{w}$.
        \end{proof}

        \begin{lemma}[\cite{Gallier_2011}]\label{toy:gen:lin2}
            Let $W, V \subset \Omega$ be two subvector spaces or
            submodules and $\vec{v},\vec{w} \in \Omega$. Then if $(W +
            \vec{w}) \cap (V + \vec{v}) \neq \emptyset$, it holds that
            
            \begin{equation}
                (W + \vec{w}) \cap (V + \vec{v}) = (W \cap V) + \vec{u}
            \end{equation}
            with $\vec{u} \in (W + \vec{w}) \cap (V + \vec{v})$.
        \end{lemma}

        \begin{proof}
            If $(W + \vec{w}) \cap (V + \vec{v}) \neq \emptyset$ then there exists a
            $\vec{u} \in (W + \vec{w}) \cap (V + \vec{v})$.
            \Cref{toy:gen:lin1} allows us to write
            \begin{align}
                \begin{split}
                    W +\vec{w} &= W + \vec{u}\\
                    V + \vec{v} &= V + \vec{u}.
                \end{split}
            \end{align}
            This means each element in $V + \vec{v}$ is of the form
            $\vec{u}_1 = \vec{v}_1 + \vec{u}$ for $\vec{v}_1 \in V$,
            and each element in $W + \vec{w}$ is of the form
            $\vec{u}_2 = \vec{w}_1 + \vec{u}$ for $\vec{w}_1 \in W$.
            Therefore $\vec{u}_1$ is in $W + \vec{w}$ if and only if $\vec{v}_1
            \in W$ and $\vec{u}_2$ is in $V + \vec{v}$ if and only if $\vec{w}_1 \in V$.
            Therefore we can conclude that $(W + \vec{w}) \cap (V + \vec{v}) = (V \cap
            W)+ \vec{u}$. $ $
        \end{proof}

        \begin{lemma}[{\cite{Catani_2017}}]\label{toy:gen:lin3}
            Let $V,W \subset \Omega$ be two subvector spaces or
            submodules then it holds that 
            \begin{equation}
                (V \oplus W)^{\perp} = V^{\perp} \cap W^{\perp}
            \end{equation}
        \end{lemma}

        \begin{proof}
            Let $\vec{a} \in (V \oplus W)^{\perp}$ and, therefore, for all vectors
            $\vec{u} \in (V \oplus W)$ it holds that $\vec{a}^{T} \vec{u} = 0$. In
            particular, this holds for $\vec{a} \in V^{\perp}$ and $\vec{a} \in
            W^{\perp}$ as $V$ and $W$ are subsets of $(V \oplus W)$. Therefore, we can conclude that
            $(V \oplus W)^{\perp} \subset V^{\perp} \cap W^{\perp}$.

            Let $\vec{b} \in V^{\perp} \cap W^{\perp}$ and let $\vec{u} \in (V
            \oplus W)$ be arbitrary. Then we find that
            \begin{equation}
                \vec{u}^T \vec{b} = \vec{u}_w^T\vec{b} + \vec{u}_v^T \vec{b} = 0
            \end{equation}
            for $\vec{u}_w \in W$ and $\vec{u}_v \in V$ such that $\vec{u}_w
            + \vec{u}_v = \vec{u}$. Therefore, $V^{\perp} \cap W^{\perp} \subset
            (V \oplus W)^{\perp}$
        \end{proof}

        \begin{lemma}[\cite{Wilding_2013,Lam_2004}]\label{toy:gen:lin4}
            Let $V \subset \Omega$ be a subset or submodule. Then it
            holds that $(V^{\perp})^{\perp} = V$.
        \end{lemma}

        \begin{proof}
            The proof for this in the case for general $d$ can be found in
            \cite{Wilding_2013}. In the case for general vector spaces 
            ($d$ prime or the continuous case) the proof can be found
            in \cite{Lam_2004}.
        \end{proof}

\subsection{Measurement as a physical process}
\label{appendix:measproofs}
Here you can find the proofs of statements used to formulate rules for measurement process.

\paragraph{Example: measuring position with a continuous 1D pointer.} 
We consider the  case where both the measured system and the pointer are continuous 1D systems, characterized by the observables $q_S, p_S, q_M$ and $p_M$.  
Note that we cannot start the pointer in the analogous of a Gaussian state, as each toy observable can only be either   fully known or completely unknown. We start instead with a pointer well-localized in position space, with $q_M=x_0$. 
Suppose that we want to measure the position of the first system, $S$. If $S$ starts in a state of well-defined position $q$, then we expect to end up in a ``classically corelated'' toy state analogous to $\ket {q}_S \ket {x_0 +q}_M$. If on the other hand $S$ starts with well-defined momentum $p$ and undefined position, we would expect the final global state to be somehow analogous to a superposition $\sum_q \alpha(q,x_0, p) \ket q_S\ket {x_0 + q}_P$. 
\Cref{thm:coherentcopyobservable} shows that in the toy theory there exists a transformation that produces a final state that is analogous to the quantum case.

\begin{restatable}{theorem}{coherentcopyposition}[Coherent copy of position]\label{thm:coherentcopyposition}
Let the epistemic state of the memory system be initialized as $q_{\text{memory}} = 0$, and let the initial epistemic state of the measured system be $(V_{\text{information}} = \langle \vec{v}_1 \rangle,\vec{v})$ with $\vec{v}_1,\vec{v} \in \mathbb{Z}_d^{2}$ or $\mathbb{R}^2$ so that the initial epistemic state of the composite system is 
\begin{equation}
    \left(\langle
\begin{pmatrix}  v^1_1 \\[6pt]  v^2_1 \\[6pt]
0 \\[6pt] 0 \end{pmatrix}, \begin{pmatrix}  0 \\[6pt] 0 \\[6pt] 1
\\[6pt] 0 \end{pmatrix} \rangle, \begin{pmatrix}  v^1 \\[6pt]  v^2 \\[6pt] 0 \\[6pt]
0 \end{pmatrix}\right).
\end{equation}
Then the transformation
\begin{equation}
    S = \begin{pmatrix}  
        1 & 0 & 0 & 0 \\[6pt] 
        0 & 1 & 0 & -1 \\[6pt]
        1 & 0 & 1 & 0 \\[6pt]
        0 & 0 & 0 & 1 \\[6pt]
    \end{pmatrix}
\end{equation}
correlates the position of the memory system and the information system. In prime and continuous dimensions tracing out the memory system results in a mixture of all possible measurement outcomes of the position of the information system. 
\end{restatable}
 
\begin{proof}
We apply the transformation to the initial state
\begin{equation}
  S:  \left(\langle
\begin{pmatrix}  v^1_1 \\[6pt]  v^2_1 \\[6pt]
0 \\[6pt] 0 \end{pmatrix}, \begin{pmatrix}  0 \\[6pt] 0 \\[6pt] 1
\\[6pt] 0 \end{pmatrix} \rangle, \begin{pmatrix}  v^1 \\[6pt]  v^2 \\[6pt] 0 \\[6pt]
0 \end{pmatrix}\right) \to  \left(\langle
\begin{pmatrix}  v^1_1 \\[6pt]  v^2_1 \\[6pt]
0 \\[6pt]  v^2_1 \\[6pt] \end{pmatrix}, \begin{pmatrix}  -1 \\[6pt] 0 \\[6pt] 1
\\[6pt] 0 \end{pmatrix} \rangle, \begin{pmatrix}  v^1 \\[6pt]  v^2 \\[6pt] v^1 \\[6pt]
0 \end{pmatrix}\right)
\end{equation}
This transformation ensures that the position of the memory system is always equal to the position of the information system. 

Furthermore, if we trace out the memory system, that is taking the marginal of the probability distribution of the ontic state over the memory system. The probability distribution before marginalisation is the uniform distribution over $V^{\perp}+\vec{v}= \text{span}((0,-1,0,1)^T,(v_1^{\perp},v_2^{\perp},v_{1}^{\perp},0)^T) +(v^1,v^2,v^1,0)^T$ where $(v_1^{\perp},v_2^{\perp})$ is the vector spanning $V_{information}^{\perp}$.
After marginalisation this results in a uniform probability distribution over $V^{\perp}+\vec{v}$ with the last two entries removed $V'^{\perp}+\vec{v'}= \text{span}((0,1)^T,(v_1^{\perp},v_2^{\perp})^T) +(v^1,v^2)$.
Therefore, if $v_1^{2} \neq 0$ the traced out state is the maximally mixed state and therefore an equal mixture of all positions of the information system. In non prime dimensions, even if  $v_1^{2} \neq 0$, the state does not need to be the maximally mixed state. On the other hand if  $v_1^{2} = 0$, the state is the state where the information system has definite position $v^1$.
\end{proof}

\paragraph{Examples.}  To obtain some intuition, let us look at two examples. Recall that this sympletic transformation acts on an initial state as $(V, \vec v) \to ((S^T)^{-1}V, S\vec v)$.
Consider the initial state analogous to $\ket\phi_S\ket{x_0}_M$, that is the product state between an arbitrary state of $S$ and a well-defined memory position $q_M = x_0$. This state is transformed as 
\begin{equation}
    \left(\left\langle
\begin{pmatrix}  v_1 \\[6pt] v_2\\[6pt]
0 \\[6pt] 0
\end{pmatrix}, \begin{pmatrix}  0 \\[6pt] 0 \\[6pt] 1 \\[6pt] 0
\end{pmatrix} \right\rangle, \begin{pmatrix}  w_1 \\[6pt] w_2 \\[6pt] x_0 \\[6pt] 0
\end{pmatrix}\right)
\quad \longrightarrow \quad 
    \left(\left\langle
\begin{pmatrix}  v_1 \\[6pt] v_2\\[6pt]
0 \\[6pt] -v_2
\end{pmatrix}, \begin{pmatrix}  1 \\[6pt] 0 \\[6pt] 1 \\[6pt] 0
\end{pmatrix} \right\rangle, \begin{pmatrix}  w_1 \\[6pt] w_2 \\[6pt]  x_0 + w_1\\[6pt] 0
\end{pmatrix}\right). 
\end{equation}
In particular, after a quick simplification we can see that it transforms the initial state  analogous to $\ket q_S \ket{x_0}_M$  as 
\begin{equation}
    \left( \left\langle \begin{pmatrix}  
        1 \\[6pt] 
        0 \\[6pt]
        0  \\[6pt]
        0 \\[6pt]
    \end{pmatrix}, \begin{pmatrix}  
        0 \\[6pt] 
        0 \\[6pt]
        1  \\[6pt]
        0 \\[6pt]
    \end{pmatrix} \right\rangle,  \begin{pmatrix}  
        q \\[6pt] 
        0 \\[6pt]
        x_0  \\[6pt]
        0 \\[6pt]
    \end{pmatrix}\right)
\qquad \longrightarrow \qquad
 \left( \left\langle \begin{pmatrix}  
        1 \\[6pt] 
        0 \\[6pt]
        0  \\[6pt]
        0 \\[6pt]
    \end{pmatrix}, \begin{pmatrix}  
        0 \\[6pt] 
        0 \\[6pt]
        1  \\[6pt]
        0 \\[6pt]
    \end{pmatrix} \right\rangle,  \begin{pmatrix}  
        q \\[6pt] 
        0 \\[6pt]
        q+x_0  \\[6pt]
        0 \\[6pt]
    \end{pmatrix}\right),
\end{equation}
that is, a state  analogous to $\ket q_S\ket{x_0+q}_M$, 
as desired (it satisfies $q_S = q$ and $q_M = q+ x_0$).
On the other hand, $U$ transforms the state analogous to $\ket p_S \ket{x_0}_M$ as 
\begin{equation}
    \left( \left\langle \begin{pmatrix}  
        0 \\[6pt] 
        1 \\[6pt]
        0  \\[6pt]
        0 \\[6pt]
    \end{pmatrix}, \begin{pmatrix}  
        0 \\[6pt] 
        0 \\[6pt]
        1  \\[6pt]
        0 \\[6pt]
    \end{pmatrix} \right\rangle,  \begin{pmatrix}  
        0 \\[6pt] 
        p \\[6pt]
        x_0  \\[6pt]
        0 \\[6pt]
    \end{pmatrix}\right)
\qquad \longrightarrow \qquad
 \left( \left\langle \begin{pmatrix}  
        0 \\[6pt] 
        1 \\[6pt]
        0  \\[6pt]
         -1 \\[6pt]
    \end{pmatrix}, \begin{pmatrix}  
        1 \\[6pt] 
        0 \\[6pt]
        1  \\[6pt]
        0 \\[6pt]
    \end{pmatrix} \right\rangle,  \begin{pmatrix}  
         0 \\[6pt] 
         p \\[6pt]
         x_0  \\[6pt]
         0 \\[6pt]
    \end{pmatrix}\right),
\end{equation} 
that is, a state such that $p_S-p_M = p$ and $q_S + q_M = x_0$. Similarly to the quantum case, the reduced states of just $S$ or just $M$ are fully mixed.

\begin{restatable}{theorem}{coherentcopyobservable}[Coherent copy of an arbitrary observable]\label{thm:coherentcopyobservable}
Let the state of the information system be $\left(W, \vec{w} \right)$, where $W$ such that it defines a state on an arbitrary amount of systems. 
Let $\vec f$ be an observable, such that  $\vec f= ((S^M)^T)^{-1} \vec q_1 $ where $\vec q_1$ the position of the first system of the information system and $S^M$ a symplectic transformation. 
Furthermore, let $\vec v = (T_{mem}^T)^{-1} \vec q$ be an observable with $T$ a symplectic transformation and $\vec q$ the position of the memory system. Let the memory system be initially in the state where the value of $\vec v$ is $0$, such that the total initial state is $(V_{tot},\vec{v}_{tot}) = \left(\langle (\vec v, 0,0, \dots)^T, (0,0,\vec w_1)^T, \dots (0,0,\vec w_k)^T\rangle, (0,0,\vec{w})^T \right)$, where $\vec w_1,...,\vec w_k$ span $W$.

Then the transformation 
\begin{equation}
 (T_{mem} \otimes \id_{mes} ) (\id_{mem} \otimes S^M_{mes}) (S \otimes \id_{mes_2,..,mes_N} ) (\id_{mem}  \otimes (S^M_{mes})^{-1}) (T_{mem}^{-1} \otimes \id_S ) )   
\end{equation}
correlates the value of $\vec v$ of the memory system and with the value of $\vec f$ of the information system. In prime and continuous dimensions tracing out the memory system results in a mixture of all possible measurement outcomes of $\vec f$ of the information system. 
\end{restatable}
\begin{proof}
The first part of the transformation $(\id_{mem}  \otimes (S^M_{mes})^{-1}) (T_{mem}^{-1} \otimes \id_S ) )$  performs a basis change, such that the observable $\vec f$ of the information system in the old basis is transformed to $\vec q_1$ in the new basis and the observable $\vec v$ of the memory system is transformed to the position of the memory system. Therefore, the case where we copy $\vec q_1$ into the position of the memory system is related by a basis change to the case where we copy $\vec{f}$ into the value of $\vec{v}$ of the memory system. By an analogous calculation as in theorem~\ref{thm:coherentcopyposition} the transformation $(S \otimes \id_{mes_2,..,mes_N} )$ correlates the position of the memory system with $\vec{q_1}$. This means that after the transformation with $(S \otimes \id_{mes_2,..,mes_N} )$ the state $(V_{tot}',\vec{v}')$ is such that $V_{tot}'$ contains a vector of the form $(1,0,-1,0,0,0,0, \dots, 0)^{T} \in V'_{tot}$ and the valuation of $(1,0,-1,0,0,0,0, \dots, 0)^{T}$ is zero. The last part of the transformation $(\id_{mem}  \otimes S^M_{mes}) (T_{mem} \otimes \id_S ) )$ transforms $(1,0,-1,0,0,0,0, \dots, 0)^{T} \to (\vec v, -\vec{f} )$ and the valuation of the new vector is still zero as the transformation is symplectic. Therefore, the above transformation correlates the value of $\vec{f}$ on the information system with the value of $\vec{v}$ on the memory system. 

As the transformation $(\id_{mem}  \otimes (S^M_{mes})^{-1}) (T_{mem}^{-1} \otimes \id_S ) )$ acts only locally on the information system and the memory system, we can first transform back the joint memory and system state, then trace out the memory system and finally only transform the information system back and get the same result as if we would have just traced out the memory system. 
To determine the marginal we must consider the probability distribution over the ontic states induced by the epistemic state. This probability distribution is just the uniform distribution over the ontic states in $U^{\perp}+ \vec u$, where $(U,\vec u)$ is the state after the measurement update. 
The vector space $(S_{M}^T \otimes T_{mem}^T U)^{\perp}$ is given by $V^{\perp}_M \oplus \langle({0,1,0,-1},\dots,0)^{T}\rangle$ where $V_{M}^{\perp}$ is spanned by the vectors $(v_i^{(1)},0, \vec v_i)$ where $\vec v_i$ are the vectors spanning $(S_M^T W)^{\perp}$ and $v_i^{(k)}$ is the $k$th entry of $\vec v_i$.
Taking the marginal of the probability distribution over the memory system gives the uniform probability distribution over $(S_M^T W)^{\perp} \oplus \langle(0,1,0,\dots,0)^{T}\rangle + S_M^{-1}\vec{w}$, as marginalisation results in removing the the entries in the vectors of the systems we marginalized over. 
Therefore, we can conclude the state of the marginal system is $(M = (S_M^T W) \cap \langle(0,1,...,0)^{T}\rangle^{\perp}, S_M^{-1}\vec{w})$. Thus all vectors in $M$ are of the form $(1,0,...)^{T}$ or $(0,0,...)^{T}$. 
This means that $M \oplus \langle (1,0,0,...,0)^{T} \rangle$ is a set of commuting observables and $(1,0,0,...,0)^{T}$ is linearly independent of $M$ if and only if $(0,1,0,...,0)^{T}$ was not already contained in $(S_M^T W)$. In this case, for each value $q \in Z_{d} $ or $\mathbb{R}$ there exists a valuation vector $\vec v_q$ such that $(1,0,0,...,0)^{T}$ has the valuation $q$ and the valuation of $M$ is constant. This means that $M$ is a mixture of states with all possible values of $q_1$ except if $\vec q_1$ was already known. 
After transforming the marginalized system back with $S_M$, the results we found for $\vec q_1$ before the transformation with $S_M$ hold now for $\vec f$.
\end{proof}

\begin{restatable}{lemma}{observabletrafoe}[Existence of symplectic transformation]\label{thm:observabletrafoe}
Consider the same setting as in theorem~\ref{thm:coherentcopyobservable}.
In continuous and prime dimensions for any observable $\vec{f}$ on the information system and any observable $\vec{v}$ on the memory system there exist symplectic transformations such that 
\begin{align}
    \vec{f} &= ((S_M)^T)^{-1} \vec q_1 \\
    \vec v &= (T^T)^{-1} \vec q.
\end{align}

\end{restatable}
\begin{proof}
Both transformations exist if we can find for any amount of systems a symplectic transformation where the first column of the transformation is given by an arbitrary vector. 
Let us call this vector $\vec w = (w_1,\dots,w_{2n})^T$.
The set of vectors 
\begin{equation}
    C = \left\{(w_1,\dots,w_{2n})^T,(0,0,w_3,w_4,\dots,w_{2n})^T,\dots, (0,\dots,0,w_{2n-1},w_{2n})^T  \right\}
\end{equation}
all commute with each other. 
Furthermore, the first vector listed in this set $(w_1,\dots,w_{2n})^T$ has symplectic inner product $1$ with the vector $(-1/w_2,0,0,\dots,0)$ if $w_2 \neq 0$ (otherwise choose $(0,1/w_1,0,\dots,0)$ and commutes with the rest of the set $C$. For the second vector the same holds for $(1/w_2,0,-1/w_4,0,\dots,0)$ if $w_4 \neq 0$ (otherwise choose $(1/w_2,0,0,1/w_3,\dots,0)$ or $(0,-1/w_1,0,1/w_3,\dots,0)$ if $w_2 = 0$). In a similar way for all vectors $\vec u$ in the set $C$ such a vector $\vec u'$  can be constructed such that it has symplectic inner product is $1$ with $\vec u$ and commutes with all other vectors in $C$.
A transformation is sympletic if its columns are such that, the first two columns have symplectic inner product $1$ with each other but commute with the rest, the same for the third and fourth column and so on for the following pairs of columns. 
Therefore, the transformation where the odd columns are the vectors from the set $C$, starting with the vector $\vec w$ and the even rows are the vectors we constructed from the previous odd column in the manner as described above is by construction symplectic and has $\vec w$ as its first column. 

\end{proof}

\begin{restatable}{theorem}{condtrafo}[Restrictions on conditional transformations]\label{thm:condtrafo}
Let $(V_S^{i},\vec{v_s^{i}})$ be the state of a system and $(V_T^{i},\vec{v_T^{i}})$ the initial state of the system to be prepared in a conditional preparation scenario. Additionally, let $S_{ST}$ be a symplectic transformation. 
Then, for $\vec{v}_{s,1}^{i},\dots, \vec{v}_{s,k}^{i}$ generating a partition of the ontic state space, i.e. 
$((V_S^{i})^{\perp} + \vec{v}_{s,1}^{i})\cup \dots \cup ((V_S^{i})^{\perp} + \vec{v}_{s,k}^{i}) = \Omega$ 
the marginals of the target system the final state $(V_T^{f} ,\vec{v}_T^{f})$ 
must be either identical or orthogonal. The number of identical states is equal for each separate type of identical states. 
\end{restatable}

\begin{proof}
After the transformation the state of the system is 
\begin{equation}
    \left(\left(S_{ST}^{-1}\right)^T \left(V_S^{i}\oplus V_T^{i}\right) ,S_{ST} \left(\vec{v_s^{i}} \oplus \vec{v_T^{i}}\right)\right).
\end{equation}
This state is allowed as $V_{S}^{i}$ only has support on the first $2 \text{ amount of systems in } S$ entries of the system and $V_{T}^{i}$ on the rest. 
The marginal of the target system after the transformation is
\begin{equation}\label{eq:condtrafo}
    \left(\left(R_{T} \left(S_{ST}^{-1}\right)^T \left(V_S^{i}\oplus V_T^{i}\right)\right) ,\Pi_{T} S_{ST} \vec{v_s^{i}} \oplus \vec{v_T^{i}}\right)
\end{equation}
where $\Pi_{T}$ the projection on the last $2 \text{ amount of systems in } T$ entries of the system and $R_T$ removes the vectors that have support in $S$.
As the transformation cannot depend on $\vec{v}_{S}^{i}$, $\vec V_{T}^{f}$ is independent of $\vec{v}_{S}^{i}$. Therefore, changing $\vec{v}_{S}^{i}$ can only change the valuation $\vec{v}_s^{f}$ of the final state. Furthermore, changing the valuation either result in orthogonal or identical states.

If the transformation is irreversible, then we can always see the transformation as a reversible transformation on a larger system and subsequently tracing out some systems. Taking the trace effectively removes vectors from $V_{T}^f$. Therefore it can make states that were orthogonal identical, but it cannot produce non orthogonal or identical states, as also in this case the transformation cannot depend on $\vec v_{S}^{i}$.

This establishes that the different marginalised states on the target system only differ in their valuations which depend only on $\vec v_{s}^{i}$. By \cref{eq:condtrafo} this dependence is a linear map. Let us call this map $C$. So, two target valuations $\vec v_{T,1}^{f}, \vec v_{T,2}^{f}$ are identical if and only if $\vec v_{T,1}^{f} - \vec v_{T,2}^{f} \in Ker(C)$. Therefore, the sets of identical marginal target states are given by $Ker(C)+\vec v_{T}^{f}$ with $\vec v_{T}^{f}$ some valuation for a target state. Therefore, each set of identical target marginal states has the same number of elements.

\changemarker{Note that all transformations in Spekkens' toy theory are given by a symplectic matrix $S$ and 
a shift vector $\vec{a} \in \Omega$, the shift vector only changes the value of the different known observables by a constant, independent of the vector $\vec v_{s}^{i}$. Therefore, adding shifts to transformations also does not help to obtain non-orthogonal states.}
\end{proof}

\begin{corollary}[Restrictions on conditional transformations: example] \label{lemma:no_conditional_transf_example}
In Spekkens' toy theory, there are no reversible transformations $U$ that implement the action
\begin{equation*} 
    \toys{00} \quad \xrightarrow{U} \quad \toys{00} , \qquad 
    \toys{10} \quad \xrightarrow{U} \quad \toys{0+}.
\end{equation*}
\end{corollary}
\begin{proof}
{This is a direct application of Theorem~\ref{thm:condtrafo}, which is the formal version of Theorem~\ref{thm:conditional_action_informal}. It also has a direct and simple proof in the stabilizer version of the toy theory;  we present that proof here for pedagogical purposes for the readers familiar with the stabilizer formalism.}

{In the stabilizer formalization of the toy theory~\cite{Pusey_2012}, for appropriate dimensions, each valid epistemic state can be isomorphically identified with  a set $S=\avg{g_1 g_2,\dots}$ of commuting or anti-commuting Pauli operators forming a valid stabilizer. Each allowed transformation on toy states corresponds to a permutation $\Pi$ on the subset of stabilizers. This permutation can be represented by a (unitary or anti-unitary) permutation matrix $V_\Pi$ that acts on each stabilizer $g\in S$ as $V_\Pi g V_\Pi^T$. 
In this case, we have }
\begin{align*}
&\toys{00} & \xrightarrow{U} & \toys{00}& , \qquad &\toys{10}& \xrightarrow{U} & \toys{0+}, \\   
    &\ket{00}\sim \avg{Z_A, Z_S} 
     & & 
    \avg{Z_A, Z_S}  
    & & 
    \ket{10}\sim \avg{-Z_A, Z_S}  
    & & 
    \ket{0+} \sim \avg{Z_A, X_S} .
\end{align*}
{This transformation $\Pi$ would have to map the stabilized states as}
\begin{align*}
    \Pi: \quad \operatorname{span}\{\mathcal{Z}_A,\mathcal{Z}_S\} 
    &\to \operatorname{span}\{\mathcal{Z}_A,\mathcal{Z}_S\}, \\
    \operatorname{span}\{-\mathcal{Z}_A,\mathcal{Z}_S\} 
    &\to 
    \operatorname{span}\{\mathcal{Z}_A,\mathcal{X}_S\}.
\end{align*}
{This means that permutation  $\Pi$  would  depend on the sign of $\mathcal{Z}_{A}$. This dependence should be explicit in the corresponding permutation matrix $V_\Pi$ (which needs to be either  unitary or anti-unitary~\cite{Pusey_2012}). However, permutation matrices cannot depend on the sign of stabilizers~\cite{Pusey_2012}, due to the linear nature of quantum theory. To see this, note that the permutation matrix would have to act on each stabilizer $g\in\{ \mathcal{Z}_{A} \otimes \id_S ,- \mathcal{Z}_{A}\otimes \id_S, \id_A\otimes \mathcal{Z}_{S}\} $ as $V_\Pi \ g\  V^\dagger_\Pi$, so if the transformation on the first state behaves as }
\begin{align*}
    \Pi: \quad \avg{\mathcal{Z}_A,\mathcal{Z}_S} 
    \to &  
    \avg{V_\Pi (\mathcal{Z}_A \otimes \id_S) V^\dagger_\Pi,\ V_\Pi (\id_A \otimes \mathcal{Z}_S) V^\dagger_\Pi} = \operatorname{span}\{\mathcal{Z}_A,\mathcal{Z}_S\} , 
\end{align*}
{then it must act on the second state as} 
\begin{align*}
     \Pi: \quad
    \operatorname{span}\{-\mathcal{Z}_A,\mathcal{Z}_S\} 
    \to\quad & 
    \operatorname{span}\{V_\Pi (- \mathcal{Z}_A \otimes \id_S)  V^\dagger_\Pi,\ V_\Pi (\id_A \otimes \mathcal{Z}_S)  V^\dagger_\Pi\} \\
    =& 
    \operatorname{span}\{ - V_\Pi ( \mathcal{Z}_A \otimes \id_S)  V^\dagger_\Pi,\ V_\Pi (\id_A \otimes \mathcal{Z}_S)  V^\dagger_\Pi\} \\
    =& \operatorname{span}\{-\mathcal{Z}_A,\mathcal{Z}_S\}
    \neq
    \operatorname{span}\{\mathcal{Z}_A,\mathcal{X}_S\}.
\end{align*}

\end{proof}

\subsection{Predictions with certainty}
\label{appendix:predproofs}
Here you can find the proofs of statements used to formulate rules for agents making predictions with certainty. The linear algebraic statements used in the proofs are formally justified in Appendix~\ref{appendix:la}.

\begin{restatable}{lemma}{paraingrcond}\label{para:ingr:cond}
    Given an epistemic state $(V,\vec{v})$ and a measurement $V_\pi$, an outcome  $\vec{v_{\pi}}$ will be measured with certainty if and only if 
    \begin{equation}
        V^{\perp} + \vec{v} \subset V_{\pi}^{\perp} + \vec{v}_{\pi}.
    \end{equation} 
    This condition is fulfilled if and only if the following two properties hold:
    \begin{enumerate}
        \item $
            V_{\pi} \subset V,
        $
        \item $
            (V^{\perp} + \vec{v}) \cap (V_{\pi}^{\perp} + \vec{v}_{\pi}) \neq \emptyset
        $
    \end{enumerate}
\end{restatable}

\begin{proof}
The probability of the measurement outcome $\vec{v}_{\pi}$ given that the system is in an epistemic state $(V,\vec{v})$ is given by the condition in~\cite{Hausmann2021}
\begin{equation}
    P(\vec{v}_{\pi}|(V,\vec{v})) = \sum_{\vec{m} \in \Omega} \xi(\vec{v}_{\pi}|\vec{m}) \mu_{(V,\vec{v})}(\vec{m}) = \frac{1}{N_{V}}\sum_{\vec{m} \in \Omega} \delta_{V_{\pi}^{\perp} +\vec{v}_{\pi}}(\vec{m}) \delta_{V^{\perp} +\vec{v}}(\vec{m}).
\end{equation} 
We are able to predict with certainty that the measurement outcome is $\vec{v}_{\pi}$ if
\begin{equation}
    (\delta_{V^{\perp} +\vec{v}}(\vec{m}) = 1)  \implies
(\delta_{V_{\pi}^{\perp} +\vec{v}_{\pi}}(\vec{m}) = 1).
\end{equation}
By the definition of $\delta_{V^{\perp} +\vec{v}}$ and $\delta_{V_{\pi}^{\perp} +\vec{v}_{\pi}}$ this condition is equivalent to~\cite{Hausmann2021}
\begin{equation}
    V^{\perp} + \vec{v} \subset V_{\pi}^{\perp} + \vec{v}_{\pi}.
\end{equation}
One might wonder why the symmetry is broken between $(V,\vec{v})$ and $(V_{\pi},\vec{v}_{\pi})$. The reason here is that $\delta_{V^{\perp} +\vec{v}}$ is normalized, while $\delta_{V_{\pi}^{\perp} +\vec{v}_{\pi}}$ is not. The
above condition can be further simplified with the following lemma.

    Let both properties be fulfilled. Then there exists a
    vector $\vec{w} \in (V^{\perp} + \vec{v}) \cap (V_{\pi}^{\perp} +
    \vec{v}_{\pi})$ and it holds that $V_{\pi} \subset V$. From the
    latter, it follows that $V_{\pi}^{\perp} \supset V^{\perp}$. This
    can be seen in the following was: if $\vec{w} \in V^{\perp}$ then
    for any $\vec{u} \in V$ it holds that $\vec{u}^T \vec{w} = 0$ and
    as $V_{\pi} \subset V$ this holds in particular for all $\vec{u}
    \in V_{\pi}$. Thus, $\vec{w} \in V_{\pi}^{\perp}$. With \cref{toy:gen:lin1}
    we can then conclude 
    \begin{align}
        V_{\pi}^{\perp} + \vec{v}_{\pi} &= V_{\pi}^{\perp} + \vec{w} \\
        V^{\perp} + \vec{v} &= V^{\perp} + \vec{w}.
    \end{align}
    Thus, it holds that 
    \begin{equation}
        V^{\perp} + \vec{v} \subset V_{\pi}^{\perp} + \vec{v}_{\pi} \neq \emptyset.
    \end{equation}

    Let 
    \begin{equation}
        V^{\perp} + \vec{v} \subset V_{\pi}^{\perp} + \vec{v}_{\pi}.
    \end{equation} be fulfilled. Then the condition 2 holds, and
    $ V^{\perp} \subset V_{\pi}^{\perp}$. Similarly as
    above, it holds that $ (V^{\perp})^{\perp} \subset
    (V_{\pi}^{\perp})^{\perp}$. Since $V \subset \Omega$ is a subset or a submodule, it is true that $(V^{\perp})^{\perp} = V$ (\cref{toy:gen:lin4}, which implies the condition 1.
\end{proof}

\begin{restatable}{lemma}{predcert} \label{pred_cert}
    Let $(V_A, v_{A = 1})$ be a measurement on subsystem $A$ and  $(V_B, v_{B = 1})$ a measurement on subsystem B. The statement ``A = 1 $\implies$ B = 1'' can be made if and only if following conditions are fulfilled
    \begin{enumerate}
        \item $
            V_{B} \subset (V_{\text{commute, A}} \oplus V_A),
        $
        \item $
            (V_{\text{commute, A}}^{\perp} + \vec{v}) \cap (V_{A}^{\perp} + \vec{v}_{A = 1}) \cap
            (V_{B}^{\perp} + \vec{v}_{B = 1}) \neq \emptyset.
        $
    \end{enumerate}
\end{restatable}

\begin{proof}
    The updated state if we measure $A = 1$ is $((V_{\text{commute, A}}
    \oplus V_a), \vec{v}')$ where $\vec{v}' \in (V_A + \vec{v}_{A = 1}) \cap
    (V_{\text{commute, A}} + \vec{v})$. Thus, \cref{para:ingr:cond} says
    that the statement ``A = 1 $\implies$ B = 1'' can be made if and
    only if 
    \begin{enumerate}
        \item $\begin{aligned}[t]
            V_{B} \subset (V_{\text{commute, A}} \oplus V_A)
        \end{aligned}$
        \item $\begin{aligned}[t] ((V_{\text{commute, A}} \oplus V_A)^{\perp} +
            \vec{v}') \cap (V_{B}^{\perp} + \vec{v}_{B = 1}) \neq \emptyset
            \end{aligned}$
    \end{enumerate}
    The second condition is equivalent to 
    \begin{equation}
        (V_{\text{commute, A}}^{\perp} + \vec{v}') \cap  (V_A^{\perp} +
        \vec{v}') \cap (V_{B}^{\perp} + \vec{v}_{B = 1}) \neq \emptyset.
    \end{equation}
    Based on~\cref{toy:gen:lin1}  
    we can conclude that the second condition is equivalent to
    \begin{equation}
        (V_{\text{commute, A}}^{\perp} + \vec{v}) \cap  (V_A^{\perp} +
        \vec{v}_{A = 1}) \cap (V_{B}^{\perp} + \vec{v}_{B = 1}) \neq \emptyset.
    \end{equation}
\end{proof}

\paragraph{Toy Bell scenario.}
The example of the Bell scenario in the generalized formalism goes as follows. Alice and Bob's shared entangled  state can be expressed as 
\begin{equation}\label{state_ex}
    (V,\vec{v}) = \left( \left\langle \begin{pmatrix}  1\\[6pt] 0 \\[6pt] 1 \\[6pt] 0
    \end{pmatrix}, \begin{pmatrix} 0\\[6pt] 1 \\[6pt] 0
    \\[6pt] -1
    \end{pmatrix} \right\rangle, \begin{pmatrix} 1\\[6pt] 0 \\[6pt] 0
        \\[6pt] 0 \end{pmatrix}  \right).
\end{equation}
Bob's $Z$ measurement corresponds to observable $\left\langle \begin{pmatrix}  0\\[6pt] 0 \\[6pt] 0 \\[6pt] 1 \end{pmatrix} \right\rangle$. If he obtains outcome $p$, he updates his description of the shared stateto 
\begin{equation}
    \left( \left\langle \begin{pmatrix}  0\\[6pt] 0 \\[6pt] 0 \\[6pt] 1
    \end{pmatrix}, \begin{pmatrix} 0\\[6pt] 1 \\[6pt] 0
    \\[6pt] -1
    \end{pmatrix} \right\rangle, \begin{pmatrix} 0\\[6pt] p \\[6pt] 0
        \\[6pt] p \end{pmatrix}  \right) = \left( \left\langle \begin{pmatrix}  0\\[6pt] 0 \\[6pt] 0 \\[6pt] 1
        \end{pmatrix}, \begin{pmatrix} 0\\[6pt] 1 \\[6pt] 0
        \\[6pt] 0
        \end{pmatrix} \right\rangle, \begin{pmatrix} 0\\[6pt] p \\[6pt] 0
            \\[6pt] p \end{pmatrix}  \right).
\end{equation}
Now, if Alice measures her system in $\left\langle \begin{pmatrix}  0\\[6pt] 1 \\[6pt] 0 \\[6pt] 0 \end{pmatrix} \right\rangle$, she will get the outcome $p$ with certainty. Thus, if Alice and Bob share the state $(V,\vec{v})$, the conclusion ``$B = p \implies A = p$'' can be made with certainty.
We can also check that the conditions of \cref{pred_cert} are fulfilled,
\begin{enumerate}
    \item $\left\langle \begin{pmatrix}  0\\[6pt] 1 \\[6pt] 0 \\[6pt]
        0 \end{pmatrix} \right\rangle \subset \left\langle \begin{pmatrix}  0\\[6pt] 0 \\[6pt] 0 \\[6pt] 1
        \end{pmatrix}\right\rangle \oplus \left\langle\begin{pmatrix} 0\\[6pt] 1 \\[6pt] 0
        \\[6pt] -1
        \end{pmatrix} \right\rangle$ , 
    \item $\begin{pmatrix}  0\\[6pt] p \\[6pt] 0 \\[6pt] p
    \end{pmatrix} \in \left( \left\langle \begin{pmatrix}  0\\[6pt] 0 \\[6pt] 0 \\[6pt] 1
    \end{pmatrix}\right\rangle^{\perp} + \begin{pmatrix} 0\\[6pt] p \\[6pt] 0
        \\[6pt] p \end{pmatrix} \right) \cap \left( \left\langle \begin{pmatrix}  0\\[6pt] 0 \\[6pt] 0 \\[6pt] 1
        \end{pmatrix} \right\rangle^{\perp} +
        \begin{pmatrix} 0\\[6pt] p \\[6pt] 0
        \\[6pt] p \end{pmatrix} \right) \cap \left(\left\langle \begin{pmatrix}  0\\[6pt] 1 \\[6pt] 0 \\[6pt] 0
        \end{pmatrix}\right\rangle^{\perp} + \begin{pmatrix} 0\\[6pt] p \\[6pt] 0
        \\[6pt] p \end{pmatrix} \right)$ .
\end{enumerate}

\subsection{No Frauchiger-Renner paradox in the toy theory}
\label{appendix:fr-proof}
\paragraph{Experimental setting.} 
We can write an arbitrary global ontic state of Alice and Bob's labs  before the measurements start as
    \begin{equation}
       \vec{o} = \begin{pmatrix} \vec{o}_R \\[6pt] \vec{o}_A \\[6pt] \vec{o}_S \\[6pt] \vec{o}_B \end{pmatrix},
    \end{equation}
    with $\vec{o}_R \in  \mathbb{Z}_d^{n_R}$ or $\mathbb{R}^{n_R}$, $\vec{o}_A \in  \mathbb{Z}_d^{n_A}$ or $\mathbb{R}^{n_A}$, $\vec{o}_S \in  \mathbb{Z}_d^{n_S}$ or $\mathbb{R}^{n_S}$, and $\vec{o}_B \in  \mathbb{Z}_d^{n_B}$ or $\mathbb{R}^{n_B}$.
 We allow for arbitrary measurements by the different agents, fixing only the range of systems measured and the order of measurements: 
    \begin{description}
        \item[Alice at t = 1] performs the following measurement
        \begin{equation}
            V_{A} = \langle \begin{pmatrix} \vec{v}_{A,1} \\[6pt] \vec{0} \end{pmatrix} , \dots ,  \begin{pmatrix} \vec{v}_{A,{k_A}} \\[6pt] \vec{0} \end{pmatrix}\rangle
        \end{equation}
        where $\vec{v}_{A} \in \mathbb{Z}_d^{n_R}$ or $\mathbb{R}^{n_R}$.
        Alice says she got the outcome ``1'' if she got the result
        $\vec{v}_{A = 1}$.

        \item[Bob at t = 2] performs the following measurement
        \begin{equation}
            V_{B} = \langle \begin{pmatrix} \vec{0} \\[6pt]\vec{v}_{B,{1}} \\[6pt] \vec{0}  \\[6pt] \end{pmatrix}, \dots , \begin{pmatrix} \vec{0} \\[6pt]\vec{v}_{B,{k_B}} \\[6pt] \vec{0}  \\[6pt] \end{pmatrix} \rangle
        \end{equation}
        where $\vec{v}_{B} \in \mathbb{Z}_d^{n_S}$ or $\mathbb{R}^{n_S}$.
        Bob says he got the outcome 1 if he got the result
        $\vec{v}_{B = 1}$.

        \item[Ursula at t = 3] performs the following measurement
        \begin{equation}
            V_{U} = \langle \begin{pmatrix} \vec{v}_{U,1} \\[6pt] \vec{0} \end{pmatrix} , \dots , \begin{pmatrix} \vec{v}_{U,k_U} \\[6pt] \vec{0} \end{pmatrix} \rangle
        \end{equation}
        where $\vec{v}_{U,1}, \vec{v}_{U,2} \in \mathbb{Z}_d^{n_A+n_R}$ or $\mathbb{R}^{n_A+n_R}$.
        Ursula says she measured ``ok'' if she got the outcome
        $\vec{v}_{U,ok}$ and ``fail'' if she got the outcome $\vec{v}_{U,fail}$.
        These two vectors need to correspond to distinct outcomes,
        i.e.\ $\vec{v}_{U,ok} -\vec{v}_{U,fail} \notin V_{U}^{\perp}$.

        \item[Wigner at t = 4] performs the following measurement
        \begin{equation}
            V_{W} = \langle \begin{pmatrix} 
             \vec{0} \\[6pt] \vec{v}_{W,1} \\[6pt] \end{pmatrix}, \dots ,\begin{pmatrix}\vec{0} \\[6pt] \vec{v}_{W,k_W} \\[6pt] \end{pmatrix} \rangle
        \end{equation}
        where $\vec{v}_{W,1}, \vec{v}_{W,2} \in \mathbb{Z}_d^{n_B+n_S}$ or $\mathbb{R}^{n_W+n_s}$.
        Wigner says he measured ``ok'' if he got the outcome
        $\vec{v}_{W,ok}$ and ``fail'' if he got the outcome $\vec{v}_{W,fail}$.
        These two vectors need to correspond to distinct outcomes,
        i.e.\ $\vec{v}_{W,ok} -\vec{v}_{W,fail} \notin V_{W}^{\perp}$.
    \end{description}

The linear algebraic statements used in the proof are formally justified in Appendix~\ref{appendix:la}.
\parathm*
\begin{proof}
    We separate the proof in four parts. First, we find necessary
    conditions on the epistemic state such that predictions can be
    made with certainty. Second, we find conditions such that the
    probability for Wigner and Ursula both get ``ok'' is non zero.
    Finally, we show that two measurement outcomes which would lead to a
    contradiction must be the same and, thus, do not lead to a
    contradiction.

    \paragraph{Epistemic state that allows predictions with certainty.} 

    We follow the chain of reasoning and determine what $V$ needs
    to fulfill to allow for predictions with certainty:

    \begin{description}
        \item[``U = ok $\implies$ B = 1''] 
        The first condition for predictions with certainty requires
        that $V_{B} \subset V_{\text{commute},U} \oplus V_U$ where
        $V_{\text{commute},U}$ denotes subset of $V$ that commutes
        with all vectors in $V_U$. This means for all vectors $\vec{v}_B \in
        V_B$ these exist vectors $\vec{v} \in V_{\text{commute},U}, \vec{v}_U \in V_U$ such that $\vec{v}_B = \vec{v} + \vec{v}_U$.
        \item[``B = 1 $\implies$ A = 1''] 
        The first condition for predictions with certainty requires
        that $V_{A} \subset V_{\text{commute},B} \oplus V_B$ where
        $V_{\text{commute},B}$ denotes subset of $V$ that commutes
        with all vectors in $V_B$. This means for all vectors $\vec{v}_A \in
        V_A$ these exist vectors $\vec{v} \in V_{\text{commute},B},\vec{v}_B \in V_B$ such that $\vec{v}_A = \vec{v} + \vec{v}_B$.
        \item[``A = 1 $\implies$ W = fail''] 
        The first condition for predictions with certainty requires
        that $V_{W} \subset V_{\text{commute},A} \oplus V_A$ where
        $V_{\text{commute},A}$ denotes subset of $V$ that commutes
        with all vectors in $V_A$. This means for all vectors $\vec{v}_W \in
        V_W$ these exist vectors $\vec{v} \in V_{\text{commute},A},\vec{v}_A \in V_A$ such that $\vec{v}_W = \vec{v} + \vec{v}_A$.
    \end{description}
    In total, this means for each $\vec{v}_W \in V_W$ there exist vectors
    $\vec{v}_1\in V_{\text{commute},A}$,$\vec{v}_2 \in V_{\text{commute},B}$,$\vec{v}_3
    \in V_{\text{commute},U}$ , $\vec{v}_A \in V_A$,$\vec{v}_B \in V_B$, and $\vec{v}_U
    \in V_U$ such that
    \begin{align}\label{para:proof:conc1}
        \begin{split}
            \vec{v}_W &= \vec{v}_1 + \vec{v}_A \\
            &= \vec{v}_1 + \vec{v}_2 + \vec{v}_B \\
            &= \vec{v}_1 + \vec{v}_2 + \vec{v}_3 + \vec{v}_U.
        \end{split}
    \end{align}

    \paragraph{$P(ok,ok)$ $\neq$ zero.}
    
    Even if the above chain of reasoning holds, the paradox only
    occurs if, additionally, the probability for Ursula and Wigner to
    both get ``ok'' $P(ok,ok)$ is not zero. As they perform a joint
    measurement their measurement is given by $V_{U} \oplus V_{W}$.
    Because Ursula and Wigner measure different subsystems we can
    choose equivalent measurement outcome vectors such that  $\vec{v}_{U,ok}
    \in V_{W}^{\perp}$ and $\vec{v}_{W,ok} \in V_{U}^{\perp}$. Thus, the
    measurement outcome $W = ok, U = ok$ has has a shift vector
    $\vec{v}_{U,ok} + \vec{v}_{W,ok}$. For $P(ok,ok) \neq 0$ to hold the state
    $(V,\vec{v})$ has to fulfill

    \begin{equation}
        ((V_{U} \oplus V_{W})^{\perp} + \vec{v}_{U,ok} + \vec{v}_{W,ok}) \cap (V^{\perp} + \vec{v}) \neq \emptyset.
    \end{equation}
    \cref{toy:gen:lin2} allows us to conclude that this
    condition is equivalent to 
    \begin{equation} \label{para:proof:cond1}
        \vec{v}_{U,ok} + \vec{v}_{W,ok} - \vec{v} \in (V_{U} \oplus V_{W})^{\perp} \oplus V^{\perp}.
    \end{equation}
    Due to $(V^{\perp})^{\perp} = V$ for a submodule or subset $V$ (\cref{toy:gen:lin4}), we can rewrite the expression $(V_{U}\oplus V_{W})^{\perp} \oplus V^{\perp}$
    
    \begin{align}
        \begin{split}
            \left(V_{U}
    \oplus V_{W}\right)^{\perp} \oplus V^{\perp} &= 
    \left( \left(  \left(V_{U} 
    \oplus V_{W}\right)^{\perp} \oplus V^{\perp}\right)^{\perp} \right)^{\perp} \\
    &= \left( \left(\left(V_{U} \oplus V_{W}\right)^{\perp}\right)^{\perp} \cap\left( V^{\perp}\right)^{\perp}\right)^{\perp} \\
    &=  \left(  \left(V_{U}
    \oplus V_{W}\right) \cap V\right)^{\perp}. 
        \end{split}
    \end{align}
    Thus, \cref{para:proof:conc1} is equivalent to 
    \begin{equation} 
        \vec{v}_{U,ok} + \vec{v}_{W,ok} - \vec{v} \in \left(  \left(V_{U}
        \oplus V_{W}\right) \cap V\right)^{\perp}.
    \end{equation}
    
    
    \paragraph{Getting the correct outcomes with certainty.} 
    We want to ensure that in the reasoning chain we cannot only
    conclude with certainty but we can also make the desired
    conclusions. For example, the previous conditions ensure that when
    Bob measures he gets an outcome with certainty
    (which can be calculated from the epistemic state), but the
    condition we consider here ensures that this outcome is $B = 1$:
    
    \begin{description}
        \item[``U = ok $\implies$ B = 1''] 
        We apply the second condition for predictions with certainty. 
        Thus, it has to hold that 
   
        \begin{equation}\label{para:proof:cond2}
            (V_{B}^{\perp} + \vec{v}_{B = 1}) \cap (V_{\text{commute}, U}^{\perp} + \vec{v}) \cap (V_U^{\perp} + \vec{v}_{U,ok}) \neq \emptyset.
        \end{equation}
        The two measurements are defined on two different subsystems.
        Therefore, we can choose, without loss of generality, $\vec{v}_{B = 1} \in
        V_U^{\perp}$ and $\vec{v}_{U,ok} \in V_B^{\perp}$. Additionally, it does not matter
        which intersection is calculated first. Therefore, we can first
        calculate the intersection $(V_{B}^{\perp} + \vec{v}_{B = 1}) \cap
        (V_U^{\perp} + \vec{v}_{U,ok})$ using~\cref{toy:gen:lin1,toy:gen:lin3}

        \begin{align}
            \begin{split}
                &(V_{B}^{\perp} + \vec{v}_{B= 1}) \cap (V_U^{\perp} + \vec{v}_{U,ok}) \\ &= (V_{B}^{\perp} + \vec{v}_{U,ok} +\vec{v}_{B= 1}) \cap (V_U^{\perp} + \vec{v}_{U,ok} + \vec{v}_{B= 1}) \\
            &= (V_{B}^{\perp}  \cap V_U^{\perp}) + \vec{v}_{B =  1} +\vec{v}_{U,ok} \\
            &= (V_{B}  \oplus V_U)^{\perp} + \vec{v}_{B= 1} +\vec{v}_{U,ok}
            \end{split}.
        \end{align}
        
        Plugging this result into \cref{para:proof:cond2} we find the
        condition
        \begin{equation}
            ((V_{B}  \oplus V_U)^{\perp} + \vec{v}_{B= 1} +\vec{v}_{U,ok}) \cap (V_{\text{commute}, U}^{\perp} + \vec{v}) \neq \emptyset.
        \end{equation}
        Due to \cref{toy:gen:lin1} this condition is equivalent
        to the condition
        \begin{equation} \label{para:proof:cond3}
            \vec{v}_{B= 1} +\vec{v}_{U,ok} - \vec{v} \in (V_{B}  \oplus V_U)^{\perp} \oplus V_{\text{commute}, U}^{\perp} = \left((V_{B}  \oplus V_U) \cap V_{\text{commute}, U}\right)^{\perp}.
        \end{equation} 

        \item[``B = 1 $\implies$ A = 1''] 
        With the same reasoning as above we find 
        \begin{equation}
            \vec{v}_{A= 1} +\vec{v}_{B = 1} - \vec{v} \in  \left((V_{B}  \oplus V_A) \cap V_{\text{commute}, B}\right)^{\perp}.
        \end{equation}
        \item[``A = 1 $\implies$ W = fail''] 
        With the same reasoning as above we find 
        \begin{equation}
            \vec{v}_{A= 1} +\vec{v}_{W,fail} - \vec{v} \in  \left((V_{A}  \oplus V_W) \cap V_{\text{commute}, A}\right)^{\perp}.
        \end{equation}
    \end{description}

    In total, the epistemic state $(V,\vec{v})$ and the measurements
    of Alice, Bob, Ursula, and Wigner need to fulfill
    \begin{enumerate}
        \item $V_{B} \subset V_{\text{commute},U} \oplus V_U$
        \item $V_{A} \subset V_{\text{commute},B} \oplus V_B$
        \item  $V_{W} \subset V_{\text{commute},A} \oplus V_A$
        \item  $\vec{v}_{U,ok} + \vec{v}_{W,ok} - \vec{v} \in \left(  \left(V_{U}
        \oplus V_{W}\right) \cap V\right)^{\perp}$
        \item $\vec{v}_{B= 1} +\vec{v}_{U,ok} - \vec{v} \in \left((V_{B}  \oplus V_U) \cap V_{\text{commute}, U}\right)^{\perp}$
        \item $\vec{v}_{A= 1} +\vec{v}_{B = 1} - \vec{v} \in  \left((V_{B}  \oplus V_A) \cap V_{\text{commute}, B}\right)^{\perp}$
        \item $\vec{v}_{A= 1} +\vec{v}_{W,fail} - \vec{v} \in  \left((V_{A}  \oplus V_W) \cap V_{\text{commute}, A}\right)^{\perp}$
    \end{enumerate}
    
    Let us assume we have found a state $(V,\vec{v})$ and measurements such
    that the above chain of reasoning holds, and $P(ok,ok) \neq 0$.
    Such a state would lead to a paradox. In the following, we show
    that such a state and measurements cannot exist. 
    
    For all $\vec{v}_W$ there exist $\vec{v}_A,\vec{v}_B$, and $
    \vec{v}_U$ be as in \cref{para:proof:conc1}. Without loss of
    generality, we can choose equivalent measurement outcomes such
    that $\vec{v}_{B = 1} \in V_U^{\perp}$ and  $\vec{v}_{U,ok} \in
    V_B^{\perp}$. Then it holds that $ \vec{v}_U - \vec{v}_B  \in V_{\text{commute},
    U}$ and $ \vec{v}_U - \vec{v}_B  \in V_B \oplus V_U$. Thus, \cref{para:proof:cond3} implies that
    
    \begin{equation}\label{para:proof:conc2}
        -\vec{v}_B^T \vec{v}_{B = 1} + \vec{v}_U^T \vec{v}_{U,ok} + \vec{v}_B^T \vec{v} -\vec{v}_{U}^T \vec{v} = 0. 
    \end{equation}
    With the same argument we can find the following conditions
    \begin{align}
        -\vec{v}_B^T \vec{v}_{B = 1} + \vec{v}_A^T \vec{v}_{A = 1} + \vec{v}_B^T \vec{v} -\vec{v}_A^T \vec{v} &= 0, \label{para:proof:conc3} \\
        -\vec{v}_W^T \vec{v}_{W,fail} + \vec{v}_A^T \vec{v}_{A = 1} - \vec{v}_A^T \vec{v} +\vec{v}_W^T \vec{v} &= 0. \label{para:proof:conc4}
    \end{align}
    Subtracting \cref{para:proof:conc3} from \cref{para:proof:conc4}
   find the condition
   \begin{equation}\label{para:proof:conc5}
    \vec{v}_W^T \vec{v}_{W,fail} -\vec{v}_W^T \vec{v} + \vec{v}_B^T \vec{v}_{B = 1} - \vec{v}_B^T \vec{v} = 0.
   \end{equation}
   Adding up \cref{para:proof:conc5} and \cref{para:proof:conc2} we
   find the condition
   \begin{equation}\label{para:proof:conc6}
    -\vec{v}_W^T \vec{v}_{W,fail} +\vec{v}_W^T \vec{v} + \vec{v}_U^T \vec{v}_{U,ok} - \vec{v}_U^T \vec{v} = 0.
\end{equation}

    If $P(ok,ok) \neq 0$ it holds that for all 
    $\vec{w} \in \left(V_{U}
    \oplus V_{W}\right) \cap V$ 
    \begin{eqnarray}
        \vec{w}^T(\vec{v}_{U,ok} + \vec{v}_{W,ok} - \vec{v}) = 0.
    \end{eqnarray} 
    In particular, it holds that $\vec{w} = \vec{v}_U - \vec{v}_W \in \left(V_{U}
    \oplus V_{W}\right) \cap V$. Thus we can conclude that the
    following condition holds
    \begin{equation}\label{para:proof:conc7}
        -\vec{v}_W^T \vec{v}_{W,ok} + \vec{v}_U^T \vec{v}_{U,ok} -\vec{v}_U^T \vec{v} +\vec{v}_W^T \vec{v} = 0.
    \end{equation}
    We can subtract \cref{para:proof:conc6} from
    \cref{para:proof:conc7} and find 
    \begin{equation}\label{para:proof:conc8}
        \vec{v}_W^T (\vec{v}_{W,ok} - \vec{v}_{W,fail}) = 0.
    \end{equation}
    Because $\vec{v}_U,\vec{v}_A,\vec{v}_B$ as in \cref{para:proof:conc1} can be found
    for all $\vec{v}_W$ it holds that $\vec{v}_{W,ok}- \vec{v}_{W,fail} \in
    V_W^{\perp}$. Thus, $\vec{v}_{W,ok}- \vec{v}_{W,fail}$ correspond to the same
    measurement outcome, as they have same valuation for any $\vec{v}_W \in
    V_W$. In summary, if there is a state and measurements such that the
    paradoxical chain  of reasoning would hold, we have shown that
    then the two measurement outcomes that would lead to a paradox
    have to be the same outcome. Therefore, no such paradoxical chain
    of reasoning is possible.
\end{proof}

\section{Review of quantum processes and experiments}
\subsection{Quantum measurements as physical processes}
\label{appendix:quantum-analogues}

The usual way to describe the measurement process in quantum theory is to start with von Neumann measurements~\cite{vonNeumann1955}, which project the system into one of the eigenstates of the observable -- such a measurement can be characterized by a set of projectors. 
However, von Neumann measurements only represent a certain class of measurements, as they contain all information about the observable. Generally, we are also interested in measurements which extract information only partially -- while they reduce the uncertainty about the observable, they don't remove it completely. These generalized measurements are known as POVMs (positive operator-valued measurements)~\footnote{They can be characterized by a generalising the set of projectors above: suppose that we pick $\{\Pi_i\}$ -- a set of $m$ operators with the only restriction $\sum_i \Pi_i^\dagger \Pi_i = \mathbb{1}$, where $m \leq n$.}.
Operationally, POVMs can be implemented by introducing another quantum system, an ancilla (which can play the part of a pointer or a memory), performing a joint unitary on both systems, and then subjecting the ancilla to a von Neumann measurement.
For example, the simplest way to measure a qubit in its computational basis with projectors $\{\proj{0}{0}_S, \proj{1}{1}_S\}$ is to perform a joint CNOT gate on the system and the memory (Figure~\ref{fig:discrete-cnot}).

Now let us consider the memory update for the case of the continuous system. For a system $\mathcal{H}_S$, we define the orthonormal basis $\{\ket{x}_S \ | \ x \in \mathbb{R}\}$ such that the states of the basis fulfill
\begin{equation}
    \hat{x}\ket{x}_S = x \ket{x}_S
\end{equation}
Let us also introduce a memory system $\mathcal{H}_M$, isomorphic to the Hilbert space $\mathcal{H}_S$. Our aim is to describe an operation which would coherently copy the state of the system $S$ to the system $M$.
We define the $CNOT_X$ gate as the transformation
\begin{equation}
    CNOT_X: |x_1\rangle_S \otimes |x_2\rangle_M \to |x_1\rangle_S \otimes |x_2 + x_1\rangle_M.
\end{equation}
We can call $CNOT_X$ a \textit{$X$-memory update}, as it is written w.r.t. the position basis. If we choose $x_2 = 0$, then the position of the first system $S$ (the one being measured) is copied into the memory system $M$. After performing the memory update on the state $|x_1\rangle_S \otimes \ket{0}_M$, the state evolves to $\ket{x_1}_S\ket{x_1}_M$. Physically, this can be implemented as the action of the Hamiltonian $H_{SM} = \hat X_S \otimes \hat P_M$ for time $t=1$, which couples two systems in the following way (we assume $\hbar=1$):
\begin{align*}
    U_{SM}(t) \left (\sum_k \alpha_k \ket{x_k}_S \otimes \ket{x}_M \right) &= \exp\left( -\frac{it}{\hbar} \left( X_S \otimes P_M \right) \right) \left( \sum_k \alpha_k \ket{x_k}_S \otimes \ket{x}_M \right) \\
    &= \sum_k \alpha_k \ket{x_k}_S \otimes \exp\left( -i x_k P_M \right) \ket{x}_M \\
    &= \sum_k \alpha_k \ket{x_k}_S \otimes \ket{x+x_k}_M.
\end{align*}
In principle, we can substitute the observable $\hat X_S$ on the system $S$ with any other physical observable $\hat A_S = \sum_k a_k \ket{a_k}\bra{a_k}_S$. In that case, the Hamiltonian takes the form $\hat H_{SM} = \hat A_S \otimes \hat P_M$, and the memory update $CNOT_X$ acts as
\begin{align*}
    CNOT_X: \sum \alpha_k\ket{a_k}\otimes \ket{x}_M \rightarrow \sum_k \alpha_k \ket{a_k} \otimes \ket{x+a_k}
\end{align*}

Suppose that the pointer $M$ has the initial position wave function $\psi_0(x)$, for instance a Gaussian wave. The interaction Hamiltonian reads $\hat H_{SM} = g\hat A_S \otimes \hat P_M$ (where we also add the factor of $g$ for quantifying the strength of the interaction), which leads to 
\begin{align*}
    \ket{\psi_k}_M &=  e^{-i t g a_k \hat P_M}  \ket{\psi_0}_M
    =
    e^{-i t g a_k \hat P_M}  \int_{-\infty}^{+\infty} dx\  \psi_0(x)\ \ket x_M
    = \int_{-\infty}^{+\infty} dx\ \psi_0(x - t\ g\  a_k)\ \ket x_M.
\end{align*}
The final global state is therefore
\begin{align*}
    \sum_k \alpha_k \ket{a_k}_S \otimes \int_{-\infty}^{+\infty} dx\ \psi_0(x - t \ g a_k)\ \ket x_M,
\end{align*}
where each outcome is correlated with a shift in the position of the pointer, and the weight of each peak corresponds to the probability of observing that outcome~\footnote{This measurement can be made sharper or weaker by tuning the parameters of the initial wave function of the pointer and the interaction time.} (Figure~\ref{fig:continuous-cnot}).   Let us just look at two examples that will be useful later to compare to the toy theory. For simplicity, suppose that we tune the interaction Hamiltonian and interaction time such that $tg=1$, that the system $S$ measured is continuous, and that we are measuring a continuous observable, $\hat A_S = \int_{-\infty}^{+\infty} dk \ a_k \proj{a_k}{a_k}_S$. In particular, let us see what happens when we measure the position ($\hat A_S = \hat X_S = \int_{-\infty}^{+\infty} dx \ x\ \proj x x_S$) of  a system $S$ that's initially in:
\begin{enumerate}
    \item A position eigenstate, $\ket {x'}_S$: the final state becomes  \begin{align}
        \ket{x'}_S \otimes \ket{\psi_{x'}} , \qquad \psi_{x'}(x) = \psi_{0} (x-x') . 
    \end{align}
    If the initial pointer state was also a position eigenstate ($\psi_0(x) = \delta(x-x_0), \ket{\psi_0}_M= \ket{x_0}_M$), we obtain 
    \begin{align}
        \ket{x'}_S \otimes \ket{x_0 + x'}.
    \end{align}
    \item A momentum eigenstate $\ket{p'}_S = (2 \pi \hbar)^{-1/2} \int_{-\infty}^{+\infty} dx' \ e^{ip'x'/\hbar} \ket{x'}_S $: the final state is entangled, 
    \begin{align}
        \frac1{\sqrt{2\pi\hbar}} \int_{-\infty}^{+\infty} dx' \ e^{ip'x'/\hbar} \ket{x'}_S \ket{\psi_{x'}}_M 
        &= 
        \frac1{\sqrt{2\pi\hbar}} \int_{-\infty}^{+\infty} dx' \ \int_{-\infty}^{+\infty} dx \ e^{ip'x'/\hbar}\ \psi_{0} (x-x')  \  \ket{x'}_S  \ket{x}_M .
    \end{align}
   Now we consider the case when the pointer is initialized in a position eigenstate $\ket{x_0}_M$ (that is, $ \psi_{0} (x) = \delta(x-x_0)$). In this case, the global state is simply 
    \begin{align}
        \frac1{\sqrt{2\pi\hbar}} \int_{-\infty}^{+\infty} dx' \ e^{ip'x'/\hbar} \ket{x'}_S \ket{x_0 + x'}_M = \int_{-\infty}^{+\infty} dp e^{-i(p'-p)x_0} \ket{p}_S \ket{p'-p}_M.
    \end{align}
\end{enumerate}

\subsection{Frauchiger-Renner thought experiment}
\label{appendix:fr-quantum}
In this appendix, we present the technical derivation of the Frauchiger-Renner paradox in quantum theory~\cite{Frauchiger_2018}. Here, we assume that the reader has basic knowledge of the postulates and notation of quantum theory. This derivation is adapted from \cite{NL2018} without conditional state preparation.

\paragraph{Systems and initial state.} We have two qubits $R$ and $S$ and two participants Alice and Bob, whose memory registers $A$ and $B$ are also modelled as one qubit each. There are two external agents Ursula and Wigner, whose quantum memories don't need to be explicitly modelled at this stage.  The initial state of $R, S, A, B$ is  a Hardy state of $R$ and $S$ \cite{Hardy1993}, and erased memories.
\begin{align}
    \ket{\psi_0}_{RSAB} &= \frac1{\sqrt3} \left(\ket0_R\ket0_S + \ket1_R \ket0_S + \ket1_R\ket0_S \right) \otimes \ket0_A \ket0_B.
\end{align}
We will describe the protocol and the evolution of the state of $RSAB$ from the perspective of Ursula and Wigner, who put Alice and Bob's labs below the Heisenberg cut in a neo-Copenhagen interpretation (that is, they model Alice and Bob's measurements as reversible physical evolutions of quantum systems). 
\begin{enumerate}[{$t= $} 1.]
    \item Alice measures $R$ in the $Z$ basis and stores the result in her memory. From Wigner's perspective, she is now entangled with $R$, 
    \begin{align}
    \ket{\psi_1}_{RASB} &= \frac1{\sqrt3} \left(\ket0_R\ket0_A\ket0_S + \ket1_R\ket1_A \ket0_S + \ket1_R\ket1_A\ket0_S \right) \otimes  \ket0_B.
\end{align}
    \textit{Notation: we changed the order of the subsystems because this will be easier later.}
    \item Bob measures $S$ in the $Z$ basis and stores the result in his memory. The global state is now a Hardy state between the two labs, 
    \begin{align}
    \ket{\psi_2}_{RASB} &= \frac1{\sqrt3} \left(\ket0_R\ket0_A\ket0_S \ket0_B + \ket1_R\ket1_A \ket0_S \ket0_B + \ket1_R\ket1_A\ket0_S \ket1_B \right)  .
\end{align}

\item Ursula measures Alice's lab ($RA$) in the Bell basis; in particular we care about outcomes with non-zero probability, which correspond to the eigenstates  
\begin{align}
    \ket{\text{ok}}_{RA} &= \frac{\ket0_R \ket0_A - \ket1_R\ket1_A}{\sqrt2}, 
\\
\ket{\text{fail}}_{RA} &= \frac{\ket0_R \ket0_A + \ket1_R\ket1_A}{\sqrt2}.
\end{align}

\item Wigner measures Bob's lab ($SB$) in the Bell basis; again we label the two eigenstates with finite probability as
\begin{align}
    \ket{\text{ok}}_{SB} &= \frac{\ket0_S \ket0_B - \ket1_S\ket1_B}{\sqrt2}, 
\\
\ket{\text{fail}}_{SB} &= \frac{\ket0_S \ket0_B + \ket1_S\ket1_B}{\sqrt2}.
\end{align}
\end{enumerate}

\paragraph{Reasoning and analysis.}
The possibility of both Ursula and Wigner getting the outcome ``ok'' is non-zero:
\begin{gather*}
P[u=w=ok]=|(\bra{\text{ok}}_{RA}\bra{\text{ok}}_{SB})\ket{\psi_2}_{RASB}|^2=\frac{1}{12}.
\end{gather*}
From now on, we post-select on this event. 
At time $t=3$, Ursula reasons about the outcome that Bob observed at $t=2$. Since we can regroup the global state before her measurement as
\begin{align}
    \ket{\psi_2}_{RASB}  = \sqrt{\frac23} \ket{\text{fail}}_{RA}\ket0_S \ket0_B + \frac{1}{\sqrt{3}}\ket{1}_R\ket{1}_A\ket1_S\ket1_B,
\end{align}
Ursula concludes that the only possibility with non-zero overlap with her observation of $\ket{\text{ok}}_{RA}$ is that Bob measured $\ket{1}_S$.  We can write this inference as ``$u=\text{ok} \implies b=1$''. 
She can further reason about what Bob, at time $t=2$ thought about Alice's outcome at time $t=1$. 
Whenever Bob observes $\ket1_S$, he can use the  same form of $\ket{\psi_2}_{RASB}$ to conclude that Alice must have measured $\ket1_R$. We can write this as ``$b=1 \implies a=1$''.
Finally, we can think about Alice's deduction about Wigner's outcome. Using the rewriting of the global state
\begin{align}
    \ket{\psi_2}_{RASB} =  \frac{1}{\sqrt{3}}\ket{0}_R\ket{0}_A \ket0_S\ket0_B + \sqrt{\frac{2}{3}} \ket1_R\ket1_A \ket{\text{fail}}_{SB},
\end{align}
we see that Alice reasons that, whenever she finds $R$ in state $\ket1_R$, then Wigner will obtain outcome ``fail'' when he measures Bob's lab. That is, ``$a=1 \implies w=\text{fail}$''.
Thus, chaining together the statements (the same reasoning that allowed the reader to solve the three hats problem), we reach an apparent contradiction: 
$$w=u=\text{ok} \implies b=1\implies a=1\implies w=\text{fail}.$$ 
That is, 
when the experiments stops with $u=w=\text{ok}$, the agents  can make \emph{deterministic}  statements about each other's reasoning and measurement results, concluding that Alice had predicted $w=\text{fail}$, hence arriving to a logical contradiction.

\newpage
\bibliographystyle{unsrtnat}


\begin{thebibliography}{36}
\providecommand{\natexlab}[1]{#1}
\providecommand{\url}[1]{\texttt{#1}}
\expandafter\ifx\csname urlstyle\endcsname\relax
  \providecommand{\doi}[1]{doi: #1}\else
  \providecommand{\doi}{doi: \begingroup \urlstyle{rm}\Url}\fi

\bibitem[Landauer(1961)]{Landauer1961}
R.~Landauer.
\newblock Irreversibility and heat generation in the computing process.
\newblock \emph{{IBM} Journal of Research and Development}, 5\penalty0
  (3):\penalty0 183--191, 1961.
\newblock \doi{10.1147/rd.53.0183}.

\bibitem[Spekkens(2005)]{Spekkens_2005}
R.~W. Spekkens.
\newblock Contextuality for preparations, transformations, and unsharp
  measurements.
\newblock \emph{Physical Review A}, 71\penalty0 (5), 2005.
\newblock \doi{10.1103/PhysRevA.71.052108}.

\bibitem[Frauchiger and Renner(2018{\natexlab{a}})]{Frauchiger_2018}
Daniela Frauchiger and Renato Renner.
\newblock Quantum theory cannot consistently describe the use of itself.
\newblock \emph{Nature Communications}, 9\penalty0 (1), 2018{\natexlab{a}}.
\newblock \doi{10.1038/s41467-018-05739-8}.

\bibitem[Brukner(2018)]{Brukner2018}
{\v{C}}.~Brukner.
\newblock A no-go theorem for observer-independent facts.
\newblock \emph{Entropy}, 20:\penalty0 350, 2018.
\newblock \doi{10.3390/e20050350}.

\bibitem[Nurgalieva and del Rio(2019{\natexlab{a}})]{Nurgalieva_2019}
Nuriya Nurgalieva and L{\'{\i}}dia del Rio.
\newblock Inadequacy of modal logic in quantum settings.
\newblock \emph{Electronic Proceedings in Theoretical Computer Science},
  287:\penalty0 267--297, 2019{\natexlab{a}}.
\newblock \doi{10.4204/EPTCS.287.16}.

\bibitem[Nurgalieva and Renner(2021)]{Nurgalieva2021}
Nuriya Nurgalieva and Renato Renner.
\newblock Testing quantum theory with thought experiments.
\newblock \emph{Contemporary Physics}, pages 1--24, 2021.
\newblock \doi{10.1080/00107514.2021.1880075}.

\bibitem[Fraser et~al.(2020)Fraser, Nurgalieva, and del Rio]{Fraser2020}
Patrick Fraser, Nuriya Nurgalieva, and Lídia del Rio.
\newblock Fitch's knowability axioms are incompatible with quantum theory.
\newblock 2020.
\newblock \doi{10.48550/ARXIV.2009.00321}.

\bibitem[Vilasini et~al.(2019)Vilasini, Nurgalieva, and del Rio]{Vilasini_2019}
V~Vilasini, Nuriya Nurgalieva, and L\'idia del Rio.
\newblock Multi-agent paradoxes beyond quantum theory.
\newblock \emph{New J. Phys.}, 21\penalty0 (11):\penalty0 113028, 2019.
\newblock \doi{10.1088/1367-2630/ab4fc4}.

\bibitem[Von~Neumann(1955)]{vonNeumann1955}
John Von~Neumann.
\newblock \emph{Mathematical foundations of quantum mechanics}.
\newblock Number~2. Princeton university press, 1955.
\newblock ISBN 9780691178561.
\newblock \doi{10.1515/9781400889921}.

\bibitem[Bennett and Shor(1998)]{Bennett1998}
Charles~H. Bennett and Peter~W. Shor.
\newblock Quantum information theory.
\newblock \emph{IEEE transactions on information theory}, 44\penalty0
  (6):\penalty0 2724--2742, 1998.
\newblock \doi{10.1109/18.720553}.

\bibitem[Pusey(2012)]{Pusey_2012}
Matthew~F. Pusey.
\newblock Stabilizer notation for spekkens' toy theory.
\newblock \emph{Foundations of Physics}, 42\penalty0 (5):\penalty0 688--708,
  2012.
\newblock \doi{10.1007/s10701-012-9639-7}.

\bibitem[Chiribella and Spekkens(2016)]{SpekkensFoundations2016}
Giulio Chiribella and Robert~W. Spekkens, editors.
\newblock \emph{Quantum Theory: Informational Foundations and Foils}, chapter
  Quasi-quantization: classical statistical theories with an epistemic
  restriction.
\newblock Springer Netherlands, 2016.
\newblock \doi{10.1007/978-94-017-7303-4}.

\bibitem[Catani and Browne(2017)]{Catani_2017}
Lorenzo Catani and Dan~E Browne.
\newblock Spekkens' toy model in all dimensions and its relationship with
  stabiliser quantum mechanics.
\newblock \emph{New Journal of Physics}, 19\penalty0 (7):\penalty0 073035,
  2017.
\newblock \doi{10.1088/1367-2630/aa781c}.

\bibitem[Coecke and Edwards(2011)]{Coecke2011}
Bob Coecke and Bill Edwards.
\newblock Spekkens's toy theory as a category of processes.
\newblock 2011.

\bibitem[Coecke et~al.(2010)Coecke, Edwards, and Spekkens]{Coecke2010}
Bob Coecke, Bill Edwards, and Robert~W. Spekkens.
\newblock Phase groups and the origin of non-locality for qubits.
\newblock \emph{Electronic Notes in Theoretical Computer Science 270 (2) (2011)
  15-36}, 2010.
\newblock \doi{10.1016/j.entcs.2011.01.021}.

\bibitem[Backens and Duman(2014)]{Backens2014}
Miriam Backens and Ali~Nabi Duman.
\newblock A complete graphical calculus for spekkens' toy bit theory.
\newblock \emph{Found. Phys. 46, 70 (2016)}, 2014.
\newblock \doi{10.1007/s10701-015-9957-7}.

\bibitem[Comfort and Kissinger(2021)]{Comfort2021}
Cole Comfort and Aleks Kissinger.
\newblock A graphical calculus for lagrangian relations.
\newblock 2021.

\bibitem[Hausmann et~al.(2021)Hausmann, Nurgalieva, and del Rio]{Hausmann2021}
Ladina Hausmann, Nuriya Nurgalieva, and Lídia del Rio.
\newblock A consolidating review of {S}pekkens' toy theory.
\newblock 2021.
\newblock \doi{10.48550/ARXIV.2105.03277}.

\bibitem[Frauchiger and Renner(2018{\natexlab{b}})]{Frauchiger2018}
Daniela Frauchiger and Renato Renner.
\newblock Quantum theory cannot consistently describe the use of itself.
\newblock \emph{Nature Communications}, 9\penalty0 (1):\penalty0 3711,
  2018{\natexlab{b}}.
\newblock ISSN 2041-1723.
\newblock \doi{10.1038/s41467-018-05739-8}.

\bibitem[Spekkens(2007)]{Spekkens07}
Robert~W. Spekkens.
\newblock Evidence for the epistemic view of quantum states: A toy theory.
\newblock \emph{Phys. Rev. A}, 75:\penalty0 032110, 2007.
\newblock \doi{10.1103/PhysRevA.75.032110}.

\bibitem[Lostaglio and Bowles(2021)]{Lostaglio2021}
Matteo Lostaglio and Joseph Bowles.
\newblock The original wigner's friend paradox within a realist toy model.
\newblock \emph{Proceedings of the Royal Society A: Mathematical, Physical and
  Engineering Sciences}, 477\penalty0 (2254), 2021.
\newblock \doi{10.1098/rspa.2021.0273}.

\bibitem[Wigner(1961)]{Wigner1961}
E.P. Wigner.
\newblock Remarks on the mind-body question.
\newblock In E.J. Good, editor, \emph{The {S}cientist {S}peculates}, pages
  284--302. London: {H}einemann, 1961.
\newblock Reprinted in E.P. Wigner, \textit{Philosophical Reflections and
  Syntheses}, pages 247--260. Springer, 1995.
  \href{https://doi.org/10.1007/978-3-642-78374-6_20}{DOI:
  10.1007/978-3-642-78374-6{\_}20}.

\bibitem[Hardy(1993)]{Hardy1993}
Lucien Hardy.
\newblock Nonlocality for two particles without inequalities for almost all
  entangled states.
\newblock \emph{Phys. Rev. Lett.}, 71:\penalty0 1665--1668, 1993.
\newblock \doi{10.1103/PhysRevLett.71.1665}.

\bibitem[Nurgalieva and del Rio(2019{\natexlab{b}})]{NL2018}
Nuriya Nurgalieva and Lídia del Rio.
\newblock Inadequacy of modal logic in quantum settings.
\newblock \emph{EPCTS}, 287:\penalty0 267--297, 2019{\natexlab{b}}.
\newblock \doi{10.4204/EPTCS.287.16}.

\bibitem[Boge(2019)]{Boge2019}
F.J. Boge.
\newblock Quantum information versus epistemic logic: an analysis of the
  {F}rauchiger{\textendash}{Renner} theorem.
\newblock \emph{Foundations of Physics}, 49:\penalty0 1143--1165, 2019.
\newblock \doi{10.1007/s10701-019-00298-4}.

\bibitem[Haddara and Cavalcanti(2022)]{Haddara2022}
Marwan Haddara and Eric~G. Cavalcanti.
\newblock A possibilistic no-go theorem on the wigner's friend paradox.
\newblock 2022.
\newblock \doi{10.48550/ARXIV.2205.12223}.

\bibitem[Bartlett et~al.(2012)Bartlett, Rudolph, and Spekkens]{Bartlett_2012}
Stephen~D. Bartlett, Terry Rudolph, and Robert~W. Spekkens.
\newblock Reconstruction of gaussian quantum mechanics from liouville mechanics
  with an epistemic restriction.
\newblock \emph{Physical Review A}, 86\penalty0 (1), 2012.
\newblock \doi{10.1103/PhysRevA.86.012103}.

\bibitem[Kochen and Specker(1975)]{Kochen1967}
Simon Kochen and E.~P. Specker.
\newblock Logical structures arising in quantum theory.
\newblock In \emph{The Logico-Algebraic Approach to Quantum Mechanics}, pages
  263--276. Springer Netherlands, 1975.
\newblock \doi{10.1007/978-94-010-1795-4_15}.

\bibitem[Cook(2004)]{Cook2004}
Roy~T Cook.
\newblock Patterns of paradox.
\newblock \emph{The Journal of Symbolic Logic}, 69\penalty0 (3):\penalty0
  767--774, 2004.
\newblock \doi{10.2178/jsl/1096901765}.

\bibitem[Abramsky et~al.(2015)Abramsky, Barbosa, Kishida, Lal, and
  Mansfield]{Abramsky15}
Samson Abramsky, Rui~Soares Barbosa, Kohei Kishida, Raymond Lal, and Shane
  Mansfield.
\newblock {Contextuality, Cohomology and Paradox}.
\newblock In Stephan Kreutzer, editor, \emph{24th EACSL Annual Conference on
  Computer Science Logic (CSL 2015)}, volume~41 of \emph{Leibniz International
  Proceedings in Informatics (LIPIcs)}, pages 211--228, Dagstuhl, Germany,
  2015. Schloss Dagstuhl--Leibniz-Zentrum fuer Informatik.
\newblock ISBN 978-3-939897-90-3.
\newblock \doi{10.4230/LIPIcs.CSL.2015.211}.

\bibitem[Pusey and Leifer(2015)]{Pusey2015}
Matthew~F. Pusey and Matthew~S. Leifer.
\newblock Logical pre- and post-selection paradoxes are proofs of
  contextuality.
\newblock 2015.
\newblock \doi{10.4204/EPTCS.195.22}.

\bibitem[Karanjai et~al.(2015)Karanjai, Cavalcanti, Bartlett, and
  Rudolph]{Karanjai2015}
Angela Karanjai, Eric~G Cavalcanti, Stephen~D Bartlett, and Terry Rudolph.
\newblock Weak values in a classical theory with an epistemic restriction.
\newblock \emph{New Journal of Physics}, 17\penalty0 (7):\penalty0 073015,
  2015.
\newblock \doi{10.1088/1367-2630/17/7/073015}.

\bibitem[Bong et~al.(2020)Bong, Utreras-Alarc{\'{o}}n, Ghafari, Liang,
  Tischler, Cavalcanti, Pryde, and Wiseman]{Cavalcanti2020}
K.-W. Bong, A.~Utreras-Alarc{\'{o}}n, F.~Ghafari, Y.-C. Liang, N.~Tischler,
  E.G. Cavalcanti, G.J. Pryde, and H.M. Wiseman.
\newblock A strong no-go theorem on the {Wigner}'s friend paradox.
\newblock \emph{Nature Physics}, 16:\penalty0 1199--1205, 2020.
\newblock \doi{10.1038/s41567-020-0990-x}.

\bibitem[Gallier(2011)]{Gallier_2011}
Jean Gallier.
\newblock Basics of {A}ffine {G}eometry.
\newblock In \emph{Texts in {A}pplied {M}athematics}, pages 7--63. Springer New
  York, 2011.
\newblock \doi{10.1007/978-1-4419-9961-0_2}.

\bibitem[Wilding et~al.(2013)Wilding, Johnson, and Kambites]{Wilding_2013}
David Wilding, Marianne Johnson, and Mark Kambites.
\newblock Exact rings and semirings.
\newblock \emph{Journal of Algebra}, 388:\penalty0 324--337, 2013.
\newblock \doi{10.1016/j.jalgebra.2013.05.005}.

\bibitem[Lam(2004)]{Lam_2004}
T.Y. Lam.
\newblock \emph{Introduction to {Q}uadratic {F}orms over {F}ields}.
\newblock American Mathematical Society, 2004.
\newblock \doi{10.1090/gsm/067}.

\end{thebibliography}


\end{document}